\documentclass[11pt]{article}

\usepackage{amsmath}
\usepackage{a4wide}
\usepackage{amssymb}
\usepackage{amsthm}
\usepackage{tikz}
\usetikzlibrary{decorations}
\usepackage{marvosym}
\usepackage{subfigure}
\usepackage[langfont=caps]{complexity}

\usetikzlibrary{decorations}
\usetikzlibrary{decorations.pathmorphing} 
\usetikzlibrary{er}
\newcommand{\comment}[1]{}

\newcommand{\URL}[1]{\textcolor{blue}{\url{#1}}}

\newtheorem{theorem}{Theorem}
\newtheorem{corollary}[theorem]{Corollary}
\newtheorem{lemma}[theorem]{Lemma}

\theoremstyle{definition}
\newtheorem{definition}{Definition}

\newclass{\coUL}{coUL}
\newclass{\modkl}{\mbox{Mod$_k$L}}
\newclass{\FL}{FL}
\newclass{\ceql}{\mbox{C$_=$L}}
\newlang{\GGR}{GGR}
\newlang{\MGGR}{3D-mGGR}
\newlang{\gB}{Bi}
\newlang{\kGonBi}{k-gon-bi}
\newlang{\GG}{k-ori-GG}
\newlang{\gReach}{Reach}
\numberwithin{equation}{section}

\newlang{\isMatch}{Decision-BPM}
\newlang{\constructMatch}{Search-BPM}
\newlang{\isUniqueMatch}{Unique-BPM}
\newclass{\FSPL}{FSPL}
\newclass{\CeqL}{\mbox{C$_=$L}}
\newclass{\modk}{\mbox{Mod$_k$L}}

\newcommand{\cut}{\mbox{\Rightscissors}}

\title{\bf{\Large{Space Complexity of Perfect Matching in Bounded Genus Bipartite Graphs}}}

\author{ {\sc Samir Datta} \footnote{Chennai Mathematical Institute, India:
    {\tt email:sdatta@cmi.ac.in}} \and  {\sc Raghav Kulkarni} \footnote{University of
    Chicago: {\tt email:raghav@cs.uchicago.edu}} \and  {\sc Raghunath
  Tewari}\footnote{University of Nebraska-Lincoln: {\tt email:rtewari@cse.unl.edu}. Research supported in part by NSF grants CCF-0830730 and CCF-0916525.} \and
{\sc N. V. Vinodchandran}\footnote{University
of Nebraska-Lincoln: {\tt email:vinod@cse.unl.edu}. Research supported in part by NSF grants CCF-0830730 and CCF-0916525.}}

\date{\today}

\begin{document}
\thispagestyle{empty}

\maketitle 

\begin{abstract}
We investigate the space complexity of certain perfect matching
problems over bipartite graphs embedded on surfaces of constant genus
(orientable or non-orientable). We show that the problems of deciding
whether such graphs have (1) a perfect matching or not and (2) a
unique perfect matching or not, are in the logspace complexity class
\SPL.  Since \SPL\ is contained in the logspace counting classes
$\oplus\L$ (in fact in \modk\ for all $k\geq 2$), \CeqL, and \PL, our
upper bound places the above-mentioned matching problems in these
counting classes as well. We also show that the search version,
computing a perfect matching, for this class of graphs is in
$\FL^{\SPL}$. Our results extend the same upper bounds for these
problems over bipartite planar graphs known earlier.

As our main technical result, we design a logspace computable and
polynomially bounded weight function which isolates a minimum weight
perfect matching in bipartite graphs embedded on surfaces of constant
genus.  We use results from algebraic topology for proving the
correctness of the weight function.
\end{abstract}

\newpage

\section{Introduction}
The {\em perfect matching} problem and its variations are one of the
most well-studied problems in theoretical computer science. Research
in understanding the inherent complexity of computational problems
related to matching has lead to important results and techniques in
complexity theory and elsewhere in theoretical computer
science. However, even after decades of research, the exact complexity
of many problems related to matching is not yet completely understood.

We investigate the {\em space complexity} of certain well studied
perfect matching problems over bipartite graphs.  We prove new uniform
space complexity upper bounds on these problems for {\em graphs
  embedded on surfaces of constant genus}. We prove our upper bounds
by solving the technical problem of `deterministically isolating' a
perfect matching for this class of graphs.

Distinguishing a single solution out of a set of solutions is a basic
algorithmic problem with many applications.  The {\em isolating lemma}
due to Mulmulay, Vazirani, and Vazirani provides a general randomized
solution to this problem.  Let ${\cal F}$ be a non-empty set system on
$U=\{1,\ldots, n\}$. The isolating lemma says, for a random weigh
function on $U$ (bounded by $n^{O(1)}$), with high probability there
is a {\em unique} set in ${\cal F}$ of minimum weight
\cite{MulmuleyEtAl87}. This lemma was originally used to give an
elegant \RNC\ algorithm for constructing a maximum matching (by
isolating a minimum weight perfect matching) in general graphs. Since
its discovery, the isolating lemma has found many applications, mostly
in discovering new randomized or non-uniform upper bounds, via
isolating minimum weight solutions \cite{MulmuleyEtAl87,
  ReinhardtAllender00,GalWigderson96,AllenderMatching99}.  Clearly,
derandomizing the isolating lemma in sufficient generality will
improve these upper bounds to their deterministic counterparts and
hence will be a major result. Unfortunately, recently it is shown that
such a derandomization will imply certain circuit lower bounds and
hence is a difficult task \cite{AM08}.

Can we bypass isolating lemma altogether and deterministically isolate
minimum weight solutions in specific situations?  Recent results
illustrate that one may be able to use the structure of specific
computational problem under consideration to achieve non-trivial
deterministic isolation.  In \cite{BTV09}, the authors used the
structure of directed paths in planar graphs to prescribe a simple
weight function that is computable deterministically in logarithmic
space with respect to which the minimum weight directed path between
any two vertices is unique. In \cite{DKR08}, the authors isolated a
perfect matching in planar bipartite graphs.  In this paper we extend
the deterministic isolation technique of \cite{DKR08} to isolate a
minimum weight perfect matching in bipartite graphs embedded on
constant genus surfaces. 

\subsection*{Our Contribution}
Let $G$ be a bipartite graph with weight function $w$ on it edges. For
an even cycle $C = e_1e_2\cdots e_{2k}$, the circulation of $C$ with
respect to $w$ is the sum $\sum_{i=1}^{2k}(-1)^{i}w(e_i)$. The main
technical contribution of the present paper can be stated
(semi-formally) as follows.

\vspace{2mm}
\noindent{\bf\em Main Technical Result}. There is a logspace matching
preserving reduction $f$, and a logspace computable and polynomially
bounded weight function $w$, so that given a bipartite graph $G$ with
a combinatorial embedding on a surface of constant genus, the
circulation of any simple cycle in $f(G)$ with respect to $w$ is
non-zero. (This implies that the minimum weight perfect matching in
$f(G)$ is unique \cite{DKR08}).
\vspace{1mm}

\comment{
We use this result to show (using known techniques) that 
{\isMatch}
for graphs embedded on a surface of constant genus is in the
complexity class {\SPL}. We also show that the problem of computing a
perfect matching if one exists (\constructMatch) is in $\FL^{\SPL}$. 
}

We use this result to establish (using known techniques) the following
new upper bounds. Refer to the next section for definitions.

\vspace{2mm}
\noindent{\bf\em New Upper Bounds.} For bipartite graphs, combinatorially embedded on
surfaces of constant genus the problems {\isMatch} and
{\isUniqueMatch} are in {\SPL}, and the problem {\constructMatch} is
in $\FL^{\SPL}$.
\vspace{1mm}

\SPL\ is a logspace complexity class that was first studied by
Allender, Reinhardt, and Zhou \cite{AllenderMatching99}. This is the
class of problems reducible to the determinant with the promise that
the determinant is either 0 or 1.  In \cite{AllenderMatching99}, the
authors show, using a non-uniform version of isolating lemma, that
perfect matching problem for general graphs is in a `non-uniform'
version of \SPL. In \cite{DKR08}, using the above-mentioned
deterministic isolation, the authors show that for planar bipartite
graphs, {\isMatch} is in fact in \SPL\ (uniformly). Recently, Hoang
showed that for graphs with polynomially many matchings, perfect
matchings and many related matching problems are in
\SPL~\cite{Hoang09}. \SPL\ is contained in logspace counting classes
such as $\modkl$ for all $k\geq 2$ (in particular in $\oplus\L$), \PL,
and $\ceql$, which are in turn contained in $\NC^2$. Thus the upper
bound of \SPL\ that we prove implies that the problems \isMatch\ and
\isUniqueMatch\ for the class of graphs we study are in these logspace
counting classes as well.

The techniques that we use in this paper can also be used to isolate
directed paths in graphs on constant genus surfaces. This shows that
the reachability problem for this class of graphs can be decided in
the unambiguous class $\UL$, extending the results of
\cite{BTV09}. But this upper bound is already known since recently
Kyn\v{c}l and Vysko\v{c}il show that reachability for bounded genus
graphs logspace reduces to reachability in planar graphs \cite{KV09}.

Matching problems over graphs of low genus have been of interest to
researchers, mainly from a parallel complexity viewpoint. The matching
problems that we consider in this paper are known to be in \NC. In
particular in \cite{KMV08}, the authors present an $\NC^2$ algorithm
for computing a perfect matching for bipartite graphs on surfaces of
$O(\log n)$ genus (readers can also find an account of known parallel
complexity upper bounds for matching problems over various classes of
graphs in their paper). However, the space complexity of matching
problems for graphs of low genus has not been investigated before. The
present paper takes a step in this direction.

\vspace{2mm}
\noindent {\bf\em Proof Outline}.  We assume that the graph $G$ is
presented as a combinatorial embedding on a surface (orientable or
non-orientable) of genus $g$, where $g$ is a constant. This is a
standard assumption when dealing with graphs on surfaces, since it is
NP-complete to check whether a graph has genus $ \leq g$
\cite{Thomassen89}. We first give a sequence of two reductions
to get, from $G$, a graph $G'$ with an embedding on a genus $g$
`polygonal schema in normal form'. These two reductions work for both
orientable and non-orientable cases. At this point we take care of the non-orientable case by reducing it to the orientable case. Once we have the embedding on an orientable polygonal schema in normal form, we
further reduce $G'$ to $G''$ where $G''$ is embedded on
a constant genus `grid graph'.  These reductions
are matching preserving, bipartiteness preserving and computable in logspace.  Finally,
for $G''$, we prescribe a set of $4g+1$ weight functions, $\mathcal{W}
= \{w_i\}_{1\leq i\leq 4g+1}$, so that for any cycle $C$ in $G''$,
there is a weight function $w_i \in \mathcal{W}$ with respect to which
the circulation of $C$ is non-zero. Since $g$ is constant, we can take
a linear combination of the elements in $\mathcal{W}$, for example
$\sum_{w_i \in \mathcal{W}}{w_i \times \left(n^c\right)^i}$ (where $n$ is the number of vertices in the grid) for some fixed
constant $c$ (say $c=4$), to get a single weight function with respect which the
circulation of any cycle is non-zero. 

The intuition behind these weight functions is as follows (for some of
the definitions, refer to later sections). The set ${\mathcal W}$ is a
disjoint union $\mathcal{W}_1 \cup \mathcal{W}_2 \cup \{w\}$ of the sets of
weight functions $\mathcal{W}_1$, $\mathcal{W}_2$, and
$\{w\}$. Consider a graph $G$ embedded on a fundamental polygon with
$2g$ sides. There are two types cycles in $G$: {\em surface separating} and {\em surface non-separating}. A basic theorem from
algebraic topology implies that a surface non-separating cycle will
intersect at least one of the sides of the polygon an odd number of
times. This leads to $2g$ weight functions in $\mathcal{W}_1$ to take
care of all the surface non-separating cycles. There are two types of
surface separating cycles: (a) ones which completely lie inside the
polygon and (b) the ones which cross some boundary. Type (a) cycles
behaves exactly like cycles in plane so the weight function $w$
designed for planar graphs works (from \cite{DKR08}). For dealing with
cycles of type (b), we first prove that if such a cycle intersects a
boundary, it should alternate between `coming in' and `going
out'. This leads to $2g$ weight functions in $\mathcal{W}_2$ which
handle all type (b) cycles.

Figure~\ref{fig:plan} gives a pictorial view of the components
involved in the proof of our main technical result.

The rest of the paper is organized as follows. In Section $2$ we give the necessary definitions and state results from earlier work, that we use in this paper. In Section $3$ we state and prove our upper bounds assuming a grid embedding. In Section $4$ we reduce the non-orientable case to the orientable one. In Section $5$ we give matching preserving, logspace reductions from a combinatorial embedding of the graph on a surface of genus $g$, to a grid embedding. In Section $6$ we add proofs of some necessarylemmas and theorems that we use to prove our results.


\begin{figure}[h!]
\centering

\begin{tikzpicture}[scale=.6,shorten >=.35mm,>=latex]

\node  (p1) at (0,0) [shape=rectangle,draw,text width = 3cm] {Combinatorial embedding of a graph on a genus $g$ orientable surface};
\node  (p2) at (9.5,0) [shape=rectangle,draw,text width = 3cm] {Combinatorial embedding on an orientable polygonal schema with $O(g)$ sides};
\node  (p3) at (19,0) [shape=rectangle,draw,text width = 3cm] {Combinatorial embedding on an orientable polygonal schema in normal form};
\node  (p4) at (19,-5) [shape=rectangle,draw,text width = 3cm] {Embedding on a ``genus $g$ grid graph''};
\node  (p5) at (9.5,-5) [shape=rectangle,draw,text width = 3cm] {Assignment of weight function $W$, w.r.t which circulations become non-zero};
\node  (p6) at (0,-5) [shape=rectangle,draw,text width = 3cm] {Minimum weight perfect matching w.r.t. $W$ is unique};

\node  (p1') at (0,6) [shape=rectangle,draw,text width = 3cm] {Combinatorial embedding of a graph on a genus $g$ non-orientable surface};
\node  (p2') at (9.5,6) [shape=rectangle,draw,text width = 3cm] {Combinatorial embedding on a non-orientable polygonal schema with $O(g)$ sides};
\node  (p3') at (19,6) [shape=rectangle,draw,text width = 3cm] {Combinatorial embedding on a non-orientable polygonal schema in normal form};

\draw [->] (p1)-- (p2) node[midway,above] {Lemma \ref{lemma:polySchema}};
\draw [->] (p2)--(p3) node[midway,above] {Theorem \ref{theorem:normal}};
\draw [->] (p3)--(p4) node[midway,right] {Lemma \ref{lemma:grid_reduction}};
\draw [->] (p4)--(p5)node[midway,above] {Theorem \ref{theorem:main}}node[midway,below] {{\scriptsize (Main Theorem)}};
\draw [->] (p1')--(p2')node[midway,above] {Lemma \ref{lemma:polySchema}};
\draw [->] (p2')--(p3')node[midway,above] {Theorem \ref{theorem:normal}};
\draw [->] (p3')--(p2)node[midway,above,sloped] {Theorem \ref{theorem:nonorientable}};
\draw [->] (p5)--(p6)node[midway,above] {Lemma \ref{lemma:uniquepm}};

\node at (0,2.4) {\bf \em Orientable case};
\node at (0,8.4) {\bf \em Non-orientable case};
\end{tikzpicture}
\label{fig:plan}
\caption{Outline of the steps. Note that all reductions are matching preserving and logspace computable.}
\end{figure}
\section{Preliminaries}

\subsection{Topological graph theory}

We introduce the necessary terminology from algebraic topology. For a
more comprehensive understanding of this topic, refer to any standard
algebraic topology book such as \cite{Massey91}.

A {\em 2-manifold} is a topological space such that every point has an
open neighborhood homeomorphic to $\mathbb{R}^2$ and two distinct
points have disjoint neighborhoods. A 2-manifold is often called a
{\em surface}. The {\em genus} of a surface $\Gamma$ is the maximum
number $g$, if there are $g$ cycles $C_1, C_2, \ldots ,C_g$ on
$\Gamma$, such that $C_i \cap C_j = \emptyset$ for all $i,j$ and
$\Gamma \setminus ( C_1 \cup C_2 \cup \ldots \cup C_g)$ is
connected. A surface is called {\em orientable} if it has two distinct
sides, else it is called {\em non-orientable}. A cycle $C$ in $\Gamma$
is said to be {\em non-separating} if there exists a path between any
two points in $\Gamma \setminus C$, else it is called {\em
  separating}.


A {\em polygonal schema} of a surface $\Gamma$, is a polygon with
$2g'$ directed sides, such that the sides of the polygon are
partitioned into $g'$ classes, each class containing exactly two sides
and glueing the two sides of each equivalence class gives the surface
$\Gamma$ (upto homeomorphism). A side in the $i$th equivalence class
is labelled $\sigma_i$ or $\bar{\sigma_i}$ depending on whether it is
directed clockwise or anti-clockwise respectively. The {\em partner} of a
side $\sigma$ is the other side in its equivalence class. By an abuse
of notation, we shall sometimes refer to the symbol of a side's partner,
as the partner of the symbol. Frequently we will denote a polygonal
schema as a linear ordering of its sides moving in a clockwise
direction, denoted by $X$. For a polygonal schema $X$, we shall refer
to any polygonal schema which is a cyclic permutation, or a reversal
of the symbols, or a complementation ($\sigma$ mapped to
$\bar{\sigma}$ and vice versa) of the symbols, as being the same as
$X$. A polygonal schema is called orientable (resp. non-orientable) if
the corresponding surface is orientable (resp. non-orientable).

\begin{definition} 
An orientable polygonal schema is said to be in {\em normal form} if it 
is in one of the following forms:
\begin{equation} \label{equation:orientSchema}
\sigma_1 \tau_1 \bar{\sigma_1}\bar{\tau_1} \sigma_2 \tau_2 \bar{\sigma_2}\bar{\tau_2} \ldots \sigma_m \tau_m \bar{\sigma_m}\bar{\tau_m} 
\end{equation}
\begin{equation}\label{equation:orientSchema2}
 \sigma \bar{\sigma}
\end{equation}

A non-orientable polygonal schema is said to be in normal form if it is of one of the 
following forms:
\begin{eqnarray}
& & \sigma \sigma X \\
& & \sigma \tau \bar{\sigma} \tau X
\end{eqnarray}
where, $X$ is a string representing an orientable schema 
in normal form (i.e. like Form~\ref{equation:orientSchema} or \ref{equation:orientSchema2} above). 
\end{definition}
We denote the polygonal schema in the normal form of a surface $\Gamma$ as $\Lambda(\Gamma)$. 
We will refer to two orientable symbols $\sigma, \tau$ which form the following 
contiguous substring: $\sigma \tau \bar{\sigma}\bar{\tau}$ as being clustered
together while a non-orientable symbol $\sigma$ which occurs like 
$\sigma \sigma$ as a contiguous subtring is said to form a {\em pair}. Thus,
in the first and third normal forms above all symbols are clustered. The first
normal form represents a connected sum of torii and the third of a projective
plane and torii. In the fourth
normal form all but one of the orientable symbols are clustered while the
only non-orientable symbol is sort of clustered with the other orientable
symbol. This form represents a connected sum of a Klein Bottle and torii.
The second normal form represents a sphere.

We next introduce the concept of $\mathbb{Z}_2$-homology. Given a
2-manifold $\Gamma$, a {\em 1-cycle} is a closed curve in
$\Gamma$. The set of 1-cycles forms an Abelian group, denoted as
$\mathcal{C}_1(\Gamma)$, under the {\em symmetric difference} operation, $\Delta$. Two 1-cycles $C_1, C_2$ are said to be
homologically equivalent if $C_1\Delta C_2$ forms the boundary of some
region in $\Gamma$. Observe that this is an equivalence relation. Then
the {\em first homology group} of $\Gamma$, $H_1(\Gamma)$, is the set
of equivalence classes of 1-cycles. In other words, if
$\mathcal{B}_1(\Gamma)$ is defined to be the subset of
$\mathcal{C}_1(\Gamma)$ that are homologically equivalent to the empty
set, then $H_1(\Gamma) =
\mathcal{C}_1(\Gamma)/\mathcal{B}_1(\Gamma)$. If $\Gamma$ is a genus
$g$ surface then $H_1(\Gamma)$ is generated by a system of $2g$
1-cycles, having only one point in common, and whose complement is
homeomorphic to a topological disk. Such a disk is also referred to as
the {\em fundamental polygon} of $\Gamma$. 

An undirected graph $G$ is said to be embedded on a surface $\Gamma$
if it can be drawn on $\Gamma$ so that no two edges cross. We assume that
the graph is given with a {\em combinatorial embedding} on a surface
of constant genus. Refer to the book by Mohar and Thomassen
\cite{MoharThomassen01} for details.  A graph $G$ is said to have {\em
  genus} $g$ if $G$ has a minimal embedding (an embedding where every
face of $G$ is homeomorphic to a disc) on a genus $g$ surface. Such an
embedding is also called a {\em 2-cell embedding}. A genus $g$ graph
is said to be orientable (non-orientable) if the surface is orientable
(non-orientable). 
\begin{definition}
The {\em polygonal schema of a graph} $G$ is a combinatorial embedding given on the polygonal schema of some surface $\Gamma$ together with the ordered set of vertices on each side of the polygon. Formally it is a tuple $(\phi , \mathcal{S})$, where $\phi$ is a cyclic ordering of the edges around a vertex and $\mathcal{S} = (S_1, S_2, \ldots , S_{2g})$ is the cyclic ordering of the directed sides of the polygon. Each $S_i$ is an ordered sequence of the vertices, from the tail to the head of the side $S_i$. Moreover every $S_i$ is paired with some other side, say $S_i^{-1}$ in $\mathcal{S}$, such that the $j$th vertex of $S_i$ (say from the tail of $S_i$) is the same as the $j$th vertex of $S_i^{-1}$ (form the tail of $S_i^{-1}$).
\end{definition}

\subsection{Complexity Theory}

For a nondeterministic machine $M$, let ${\it acc}_{M}(x)$ and ${\it
  rej}_{M}(x)$ denote the number of accepting computations and the
number of rejecting computations respectively. Denote ${\it
  gap}_{M}(x) = {\it acc}_{M}(x)-{\it rej}_{M}(x)$.

\begin{definition}
A language $L$ is in $\SPL$ if there exists a logspace bounded
nondeterministic machine $M$ so that for all inputs $x$, ${\it gap}_{M}(x)\in \{0 ,1\}$ and $x\in L$ if and only
if ${\it gap}_{M}(x) = 1$. $\FL^{\SPL}$ is the class
of functions computed by a logspace machine with an {\SPL}
oracle. {\UL} is the class of languages $L$,
decided by a nondeterministic logspace machine (say $M$), such that
for every string in $L$, $M$ has exactly one accepting path and for a
string not in $L$, $M$ has no accepting path.
\end{definition}

Alternatively, we can define {\SPL} as the class of problems logspace
reducible to the problem of checking whether the determinant of a matrix is
$0$ or not under the promise that the determinant is either $0$ or
$1$.  For definitions of other complexity classes refer to any
standard textbooks such as \cite{AroraBarak09,Vollmer99}.  All reductions
discussed in this paper are logspace reductions.


Given an undirected graph $G=(V,E)$, a {\em matching} $M$ is a subset
of $E$ such that no two edges in $M$ have a vertex in common. A {\em
  maximum matching} is a matching of maximum cardinality. $M$ is said
to be a perfect matching if every vertex is an endpoint of some edge
in $M$. 

\begin{definition}
We define the following computational problems related to matching:
\begin{itemize}

\vspace{-2mm}
\item[-]
{\isMatch} : Given a bipartite graph $G$, checking if $G$ has a perfect matching.

\vspace{-2mm}
\item[-]
{\constructMatch}: Given a bipartite graph $G$, constructing a perfect matching, if one exists.

\vspace{-2mm}
\item[-]
{\isUniqueMatch}: Given a bipartite graph $G$, checking if $G$ has a unique perfect matching.
\end{itemize}
\end{definition}

\subsection{Necessary Prior Results}

\begin{lemma}[\cite{DKR08}] 
\label{lemma:uniquepm}For any bipartite graph $G$ and a weight function $w$, if all circulations of $G$ are non-zero, then $G$ has a unique minimum weight perfect matching.
\end{lemma}

\begin{lemma} [\cite{AllenderMatching99}] 
\label{lemma:spl}For any weighted graph $G$ assume that the minimum weight perfect matching in $G$ is unique and also for any subset of edges $E' \subseteq E$, the minimum weight perfect matching in $G\setminus E'$ is also unique. Then deciding if $G$ has a perfect matching is in {\SPL}. Moreover, computing the perfect matching (in case it exists) is in $\FL^{\SPL}$.
\end{lemma}
\begin{proof}[Sketch of proof]
Let $w_{max}$ and $w_{min}$ be the maximum and minimum possible weights respectively, that an edge in $G$ can get. Then any perfect matching in $G$ will have a weight from the set $W = \{k : k \in \mathbb{Z}, n\cdot w_{min} \leq k \leq n\cdot w_{max} \}$. Similar to \cite{AllenderMatching99}, there exists a {\GapL} function $f$, such that for some value of $k \in W$, $|f(G,k)| = 1$ if $G$ has a perfect matching of weight $k$, else $f(G,k)=0$ for all values of $k$. Note that in \cite{AllenderMatching99} the authors actually give a $\GapL/\poly$ function since the weight function for the graphs (which are unweighted to begin with) are required as an advice in their {\GapL} machine. Here we consider weighted graphs, thus eliminating the need for any advice. Now consider the function 
\[
g(G) = 1- \prod_k \left(1-(f(G,k))^2\right). 
\]
By definition, $g(G) =1$ if $G$ has a perfect matching, else it is $0$.

To compute a perfect matching in $G$, we will construct a logspace transducer that makes several queries to the function $f$ defined above. For a graph $G'$ having a unique minimum weight perfect matching (say $M'$), the weight of $M'$ can be computed by iteratively querying the function $f(G',k)$ for values of $k \in W$ in an increasing order, starting from $n\cdot w_{min}$. The value $k$, for which the function outputs a non-zero value for the first time, is the weight of $M'$. We denote this weight by $w_{G'}$. First compute $w_G$. For an $e$ in $G$, define the graph $G^{-e} = G \setminus \{e\}$. Now compute $w_{G^{-e}}$ for every edge $e$ in $G$. Output the edges $e$ for which $w_{G^{-e}} > w_{G}$. The set of outputted edges comprise a perfect matching (in fact the minimum weight perfect matching) because deleting an edge in this set had increased the weight of the minimum weight perfect matching in the resulting graph. 
\end{proof}

\section{Embedding on a Grid}


\begin{theorem} \label{theorem:main_reduction}
Given a 2-cell combinatorial embedding of a graph $G$ of constant genus, there is a logspace transducer that constructs a graph $G' \in \GG$, such that, there is a perfect matching in $G$ iff there is a perfect matching in $G'$. Moreover, given a perfect matching $M'$ in $G'$, in logspace one can construct a perfect matching $M$ in $G$.  
\end{theorem}
\begin{proof}
Using Corollary  \ref{cor:PS_reduction} reduce $G$ to a graph $G_1$ that has an embedding on the polygonal schema in the normal form. If the schema is non-orientable, then by applying Theorem \ref{theorem:nonorientable} we get a graph $G_2$ along with its embedding on an orientable polygonal schema (need not be in the normal form). Again by Corollary \ref{cor:PS_reduction}, we reduce it to a graph on a polygonal schema in the normal form. Finally we apply Lemma \ref{lemma:grid_reduction}, we get the desired graph.
\end{proof}

\subsubsection{Combinatorial Embedding to a Polygonal Schema}
\begin{lemma} [\cite{ABCDR09}]
\label{lemma:snsCycle}
 Let $G$ be a graph embedded on a surface, and let $T$ be a spanning 
tree of $G$. Then there is an edge $e \in E(G)$ such that $T \cup \{e\}$
 contains a non-separating cycle. 
\end{lemma}
Notice that in \cite{ABCDR09} the graph was required to be embedded on an
orientable surface but the proof did not use this requirement.

\begin{definition} Given a cycle (or path) $C$ in an embedded
graph $G$, define by $G \cut C$ the graph constructed by ``cutting''
the edges incident on the cycle from the right. In other words, 
the neighbors of $u \in C$ (which are not on the cycle) can be partitioned 
into two sets, arbitrarily called left and right. For every neighbbor $v$
of $u$ which lies to the right of $C$, cut the edge $(u,v)$ into two
pieces $(u,x_{uv})$ and $(y_{uv}, v)$ where $x_{uv}, y_{uv}$ are (new) spurious
vertices.  We add spurious edges between consecutive spurious vertices along 
the cut and label all the newly formed spurious edges with the label $L_C$ 
along the left set and $L_C^{-1}$ along the right set.
(see Figure \ref{fig:cut}). 

\begin{figure}[h]
\centering
\begin{tikzpicture}[scale=.85,shorten >=.35mm,>=latex]
[decoration=zigzag]
\draw [thick,-] (-1,0) -- (0,0) -- (2,0) -- (4,0) -- (5,0);
\draw [thick,dashed] (-2,0) -- (-1,0);
\draw [thick,dashed] (5,0) -- (6,0);

\draw [-] (0,0) -- (-.5,1);
\draw [-] (0,0) -- (.5,1);
\draw [-] (0,0) -- (0,1);
\draw [-] (0,0) -- (-.5,-1);
\draw [-] (0,0) -- (.5,-1);
\draw [-] (2,0) -- (2,1);
\draw [-] (2,0) -- (1.5,-1);
\draw [-] (2,0) -- (2.5,-1);
\draw [-] (2,0) -- (2,-1);
\draw [-] (4,0) -- (3.5,1);
\draw [-] (4,0) -- (4.5,1);
\draw [-] (4,0) -- (4,1);
\draw [-] (4,0) -- (3.25,-1);
\draw [-] (4,0) -- (3.75,-1);
\draw [-] (4,0) -- (4.25,-1);
\draw [-] (4,0) -- (4.75,-1);
\draw [fill, color=black!90] (0,0) circle(0.5 mm);
\draw [fill, color=black!90] (2,0) circle(0.5 mm);
\draw [fill, color=black!90] (4,0) circle(0.5 mm);

\draw [fill, color=black!90] (-.5,1) circle(0.5 mm);
\draw [fill, color=black!90] (0,1) circle(0.5 mm);
\draw [fill, color=black!90] (.5,1) circle(0.5 mm);
\draw [fill, color=black!90] (2,1) circle(0.5 mm);
\draw [fill, color=black!90] (3.5,1) circle(0.5 mm);
\draw [fill, color=black!90] (4,1) circle(0.5 mm);
\draw [fill, color=black!90] (4.5,1) circle(0.5 mm);

\draw [fill, color=black!90] (-.5,-1) circle(0.5 mm);
\draw [fill, color=black!90] (.5,-1) circle(0.5 mm);
\draw [fill, color=black!90] (1.5,-1) circle(0.5 mm);
\draw [fill, color=black!90] (2,-1) circle(0.5 mm);
\draw [fill, color=black!90] (2.5,-1) circle(0.5 mm);
\draw [fill, color=black!90] (3.25,-1) circle(0.5 mm);
\draw [fill, color=black!90] (3.75,-1) circle(0.5 mm);
\draw [fill, color=black!90] (4.25,-1) circle(0.5 mm);
\draw [fill, color=black!90] (4.75,-1) circle(0.5 mm);

\node at (-.3,.2) {$X$};
\node at (1.7,.2) {$Y$};
\node at (3.7,.2) {$Z$};
\node [below] at (-.5,-1) {$X_1$};
\node [below] at (.5,-1) {$X_2$};
\node [below] at (1.5,-1) {$Y_1$};
\node [below] at (2,-1) {$Y_2$};
\node [below] at (2.5,-1) {$Y_3$};
\node [below] at (3.25,-1) {$Z_1$};
\node [below] at (3.75,-1) {$Z_2$};
\node [below] at (4.25,-1) {$Z_3$};
\node [below] at (4.75,-1) {$Z_4$};

\node [below] at (-1.5,0) {$C$};
\node at (2,-2) {(a)};
\draw [thick,-] (-1,-5) -- (0,-5) -- (2,-5) -- (4,-5) -- (5,-5);
\draw [thick,dashed] (-2,-5) -- (-1,-5);
\draw [thick,dashed] (5,-5) -- (6,-5);
\draw [very thick,dotted] (-1,-6) -- (5.25,-6);
\draw [very thick,dotted] (-1,-7.5) -- (5.25,-7.5);
\draw [-] (0,-5) -- (-.5,-4);
\draw [-] (0,-5) -- (.5,-4);
\draw [-] (0,-5) -- (0,-4);
\draw [-] (0,-5) -- (-.5,-6);
\draw [-] (0,-5) -- (.5,-6);
\draw [-] (-.5,-7.5) -- (-.5,-8.5);
\draw [-] (.5,-7.5) -- (.5,-8.5);
\draw [-] (2,-5) -- (2,-4);
\draw [-] (2,-5) -- (1.5,-6);
\draw [-] (2,-5) -- (2.5,-6);
\draw [-] (2,-5) -- (2,-6);
\draw [-] (1.5,-8.5) -- (1.5,-7.5);
\draw [-] (2.5,-8.5) -- (2.5,-7.5);
\draw [-] (2,-8.5) -- (2,-7.5);
\draw [-] (4,-5) -- (3.5,-4);
\draw [-] (4,-5) -- (4.5,-4);
\draw [-] (4,-5) -- (4,-4);
\draw [-] (4,-5) -- (3.25,-6);
\draw [-] (4,-5) -- (3.75,-6);
\draw [-] (4,-5) -- (4.25,-6);
\draw [-] (4,-5) -- (4.75,-6);
\draw [-] (3.25,-8.5) -- (3.25,-7.5);
\draw [-] (3.75,-8.5) -- (3.75,-7.5);
\draw [-] (4.25,-8.5) -- (4.25,-7.5);
\draw [-] (4.75,-8.5) -- (4.75,-7.5);
\draw [fill, color=black!90] (0,-5) circle(0.5 mm);
\draw [fill, color=black!90] (2,-5) circle(0.5 mm);
\draw [fill, color=black!90] (4,-5) circle(0.5 mm);

\draw [fill, color=black!90] (-.5,-4) circle(0.5 mm);
\draw [fill, color=black!90] (0,-4) circle(0.5 mm);
\draw [fill, color=black!90] (.5,-4) circle(0.5 mm);
\draw [fill, color=black!90] (2,-4) circle(0.5 mm);
\draw [fill, color=black!90] (3.5,-4) circle(0.5 mm);
\draw [fill, color=black!90] (4,-4) circle(0.5 mm);
\draw [fill, color=black!90] (4.5,-4) circle(0.5 mm);

\foreach \x in {-.5,.5,1.5,2,2.5,3.25,3.75,4.25,4.75}
{
	\foreach \y in {-6,-7.5,-8.5}
	{
		\draw [fill, color=black!90] (\x,\y) circle(0.5 mm);
	}	
}

\node at (-.3,-4.8) {$X$};
\node at (1.7,-4.8) {$Y$};
\node at (3.7,-4.8) {$Z$};
\node [below] at (-.5,-6) {$X_1'$};
\node [below] at (.5,-6) {$X_2'$};
\node [below] at (1.5,-6) {$Y_1'$};
\node [below] at (2,-6) {$Y_2'$};
\node [below] at (2.5,-6) {$Y_3'$};
\node [below] at (3.25,-6) {$Z_1'$};
\node [below] at (3.75,-6) {$Z_2'$};
\node [below] at (4.25,-6) {$Z_3'$};
\node [below] at (4.75,-6) {$Z_4'$};
\node [above] at (-.5,-7.5) {$X_1''$};
\node [above] at (.5,-7.5) {$X_2''$};
\node [above] at (1.5,-7.5) {$Y_1''$};
\node [above] at (2,-7.5) {$Y_2''$};
\node [above] at (2.5,-7.5) {$Y_3''$};
\node [above] at (3.25,-7.5) {$Z_1''$};
\node [above] at (3.75,-7.5) {$Z_2''$};
\node [above] at (4.25,-7.5) {$Z_3''$};
\node [above] at (4.75,-7.5) {$Z_4''$};
\node [below] at (-.5,-8.5) {$X_1$};
\node [below] at (.5,-8.5) {$X_2$};
\node [below] at (1.5,-8.5) {$Y_1$};
\node [below] at (2,-8.5) {$Y_2$};
\node [below] at (2.5,-8.5) {$Y_3$};
\node [below] at (3.25,-8.5) {$Z_1$};
\node [below] at (3.75,-8.5) {$Z_2$};
\node [below] at (4.25,-8.5) {$Z_3$};
\node [below] at (4.75,-8.5) {$Z_4$};

\node [below] at (-1.5,-5) {$C$};
\node at (2,-9.5) {(b)};
\end{tikzpicture}
\caption{An example of the cut operation $\cut$, cutting graph $G$ along cycle (or path) $C$. (a) Part of graph $G$ and cycle $C$. (b) Part of the resulting graph $G \cut C$, with the dotted lines representing the spurious edges.}
\label{fig:cut}
\end{figure}
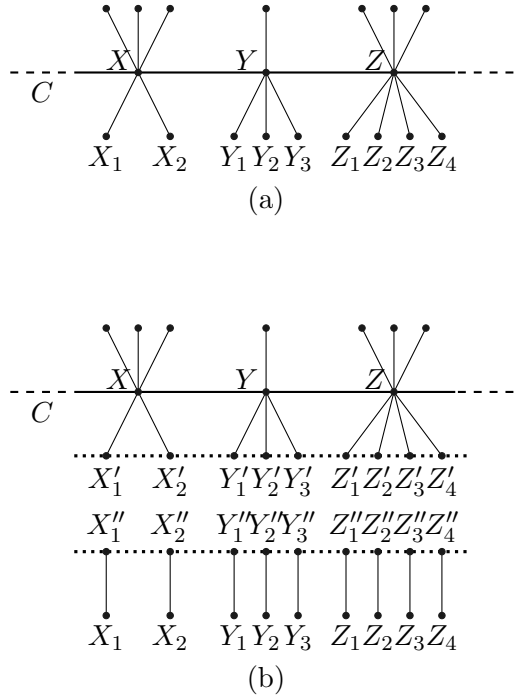

Also, if $C$ is a path, its endpoints will
lie on two paths. Consider the first path - if the two edges on either side
of $C$ on this path have the same label $L_1$. This can be broken into two cases - firstly, if the left and right side of this endpoint are the same (in other words, the path is a cycle). In this case, we just keep the same label $L_1$. When the left and right side of this endpoint are distinct, we will need to split the
label into two or three new labels as detailed below and similarly
for the other path and common label $L_2$. We will only describe the case
when $L_1,L_2$ are both defined - the other cases are similar and simpler.

 First assume that $L_1 \neq L_2$
and $L_1 \neq  L_2^{-1}$. Then we will split remove labels $L_1,L_2$
and replace them by four new labels say $L'_{1,C}, L''_{1,C}$ and
$L'_{2,C}, L''_{2,C}$, respectively for the two sides of the intersection.
If, on the other hand, $L_1$ is the same as $L_2$ or its inverse - then there
are two subcases. Firstly, if the path $C$ is between two copies of the same
vertex then we replace $L_1$ by two new labels $L'_{1,C}, L''_{1,C}$
one for either side of the cut. $L_2$ being a copy or an inverse copy
of $L_1$ splits automatically. The second case is if $C$ is between two distinct
points on two copies or inverse copies. Then we split $L_1$ into three
parts according to the two points.
The rotation system is modified appropriately. We illustrate this in Figure \ref{fig:cut_operation_2}.
\end{definition}

Notice that in the process of cutting, for every new label $L_C$ we are adding 
at most $4$ new labels.

\begin{figure}[h]
\centering

\begin{tikzpicture}[scale=.6,shorten >=.35mm,>=latex]

\foreach \y in {0}{
\draw [] [-] (0,\y-1) -- (0,\y+1);
\draw [][-] (3,\y-1) -- (3,\y+1);
\draw [][->] (0,\y) -- (3,\y);

\node [left] at (0,\y-.5) {$L_1$};
\node [left] at (0,\y+.5) {$L_1$};
\node [right] at (3,\y-.5) {$L_2$};
\node [right] at (3,\y+.5) {$L_2$};
\node [below] at (1.5,\y) {$C$};

\draw [][->] (5,\y) -- (8,\y);
\node [above] at (6.5,\y) {reduction};

\draw [] [-] (10,\y-1.5) -- (10,\y-.5);
\draw [] [-] (10,\y+1.5) -- (10,\y+.5);
\draw [][-] (13,\y-1.5) -- (13,\y-.5);
\draw [][-] (13,\y+1.5) -- (13,\y+.5);
\draw [][->] (10,\y+.5) -- (13,\y+.5);
\draw [][->] (10,\y-.5) -- (13,\y-.5);

\node [left] at (10,\y-1) {$L'_{1,C}$};
\node [left] at (10,\y+1) {$L''_{1,C}$};
\node [right] at (13,\y-1) {$L'_{2,C}$};
\node [right] at (13,\y+1) {$L''_{2,C}$};
\node [below] at (11.5,\y-0.5) {$L_C^{-1}$};
\node [above] at (11.5,\y+0.5) {$L_C$};

\node at (6.5,\y-2) {(a)};
}
\foreach \y in {-4}{
\draw [] [-] (0,\y-1) -- (0,\y+1);
\draw [][-] (3,\y-1) -- (3,\y+1);
\draw [][->] (0,\y) -- (3,\y);

\node [left] at (0,\y-.7) {$L_1$};
\node [left] at (0,\y+.7) {$L_1$};
\node [right] at (3,\y-.7) {$L_1$};
\node [right] at (3,\y+.7) {$L_1$};
\node [below] at (1.5,\y) {$C$};
\node [left] at (0,\y) {$v$};
\node [right] at (3,\y) {$v$};
\draw [fill, color=black!90] (0,\y) circle(0.5 mm);
\draw [fill, color=black!90] (3,\y) circle(0.5 mm);

\draw [][->] (5,\y) -- (8,\y);
\node [above] at (6.5,\y) {reduction};

\draw [] [-] (10,\y-1.5) -- (10,\y-.5);
\draw [] [-] (10,\y+1.5) -- (10,\y+.5);
\draw [][-] (13,\y-1.5) -- (13,\y-.5);
\draw [][-] (13,\y+1.5) -- (13,\y+.5);
\draw [][->] (10,\y+.5) -- (13,\y+.5);
\draw [][->] (10,\y-.5) -- (13,\y-.5);

\node [left] at (10,\y-1) {$L'_{1,C}$};
\node [left] at (10,\y+1) {$L''_{1,C}$};
\node [right] at (13,\y-1) {$L'_{1,C}$};
\node [right] at (13,\y+1) {$L''_{1,C}$};
\node [below] at (11.5,\y-0.5) {$L_C^{-1}$};
\node [above] at (11.5,\y+0.5) {$L_C$};
\node [left] at (10,\y+.3) {$v$};
\node [right] at (13,\y+.3) {$v$};
\node [left] at (10,\y-.3) {$v$};
\node [right] at (13,\y-.3) {$v$};
\draw [fill, color=black!90] (10,\y+.5) circle(0.5 mm);
\draw [fill, color=black!90] (13,\y+.5) circle(0.5 mm);
\draw [fill, color=black!90] (10,\y-.5) circle(0.5 mm);
\draw [fill, color=black!90] (13,\y-.5) circle(0.5 mm);

\node at (6.5,\y-2) {(b)};
}

\foreach \y in {-10}{
\draw [] [-] (0,\y-2) -- (0,\y+1);
\draw [][-] (3,\y-2) -- (3,\y+1);
\draw [][->] (0,\y-1) -- (3,\y);

\node [left] at (0,\y-1.5) {$L_1$};
\node [left] at (0,\y+.7) {$L_1$};
\node [right] at (3,\y-1.5) {$L_1$};
\node [right] at (3,\y+.7) {$L_1$};
\node [below] at (1.5,\y-.4) {$C$};
\node [left] at (0,\y) {$w$};
\node [right] at (3,\y) {$w$};
\node [left] at (0,\y-1) {$v$};
\node [right] at (3,\y-1) {$v$};
\draw [fill, color=black!90] (0,\y) circle(0.5 mm);
\draw [fill, color=black!90] (3,\y) circle(0.5 mm);
\draw [fill, color=black!90] (0,\y-1) circle(0.5 mm);
\draw [fill, color=black!90] (3,\y-1) circle(0.5 mm);

\draw [][->] (5,\y) -- (8,\y);
\node [above] at (6.5,\y) {reduction};

\draw [] [-] (10,\y-1.5) -- (10,\y-.5);
\draw [] [-] (10,\y+2.5) -- (10,\y+.5);
\draw [][-] (13,\y-2.5) -- (13,\y-.5);
\draw [][-] (13,\y+1.5) -- (13,\y+.5);
\draw [][->] (10,\y+.5) -- (13,\y+.5);
\draw [][->] (10,\y-.5) -- (13,\y-.5);

\node [left] at (10,\y-1) {$L'_{1,C}$};
\node [left] at (10,\y+1) {$L''_{1,C}$};
\node [right] at (13,\y-1) {$L''_{1,C}$};
\node [right] at (13,\y+1) {$L'''_{1,C}$};
\node [left] at (10,\y+2) {$L'''_{1,C}$};
\node [right] at (13,\y-2) {$L'_{1,C}$};
\node [below] at (11.5,\y-0.5) {$L_C^{-1}$};
\node [above] at (11.5,\y+0.5) {$L_C$};
\node [left] at (10,\y+.3) {$v$};
\node [right] at (13,\y+.3) {$w$};
\node [left] at (10,\y-.3) {$v$};
\node [right] at (13,\y-.3) {$w$};
\node [left] at (13,\y-1.5) {$v$};
\node [right] at (10,\y+1.5) {$w$};
\draw [fill, color=black!90] (10,\y+.5) circle(0.5 mm);
\draw [fill, color=black!90] (13,\y+.5) circle(0.5 mm);
\draw [fill, color=black!90] (10,\y-.5) circle(0.5 mm);
\draw [fill, color=black!90] (13,\y-.5) circle(0.5 mm);
\draw [fill, color=black!90] (10,\y+1.5) circle(0.5 mm);
\draw [fill, color=black!90] (13,\y-1.5) circle(0.5 mm);

\node at (6.5,\y-3) {(c)};
}
\end{tikzpicture}
\caption{Cutting along a path $C$ when (a) $L_1 \ne L_2$ and $L_1 \ne L^{-1}_2$, (b) $L_1 = L_2$ or $L_1 = L_2^{-1}$ and $C$ is between copies of the same vertex $v$, and (c) $L_1 = L_2$ or $L_1 = L_2^{-1}$ and $C$ is between distinct vertices $v$ and $w$.}

\label{fig:cut_operation_2}
\end{figure}
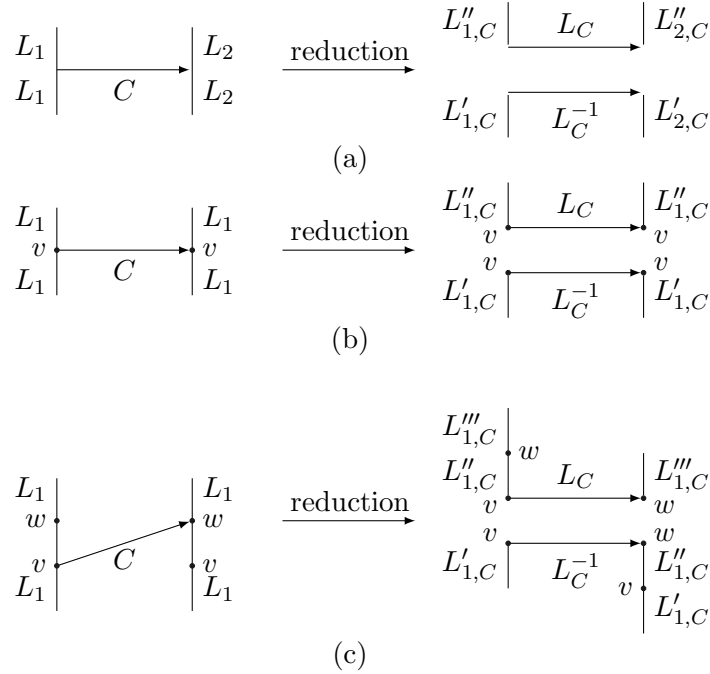

Given a graph $G_i$ embedded on a surface, potentially with
spurious edges, we can find $C_{i+1}$, a non-separating cycle 
(which does not use a spurious edge) by invoking Lemma~\ref{lemma:snsCycle}.
Define $G_{i+1}$ to be $G_i \cut C_{i+1}$.

Starting with $G_0 = G$ of genus $g$ and repeating the above operation at 
most $g$ times, we get a planar graph $H$ with at most $2g$ spurious 
faces (which consist of spurious vertices and edges).

Now find a spanning tree of this graph which does not use a spurious edge - 
that such a tree exists follows from noticing that the graph without spurious
edges is still connected.  Find a tree path connecting any two
spurious faces. Cut along this path to combine the two spurious faces into one 
larger spurious face. Repeat the operation till all the spurious faces are 
merged into one spurious face and re-embed the planar graph so that it forms 
the external face.

It is easy to see that the procedure above can be performed in logspace, provided
that $g$ is constant. Thus we have sketched the proof of the following:
\begin{lemma}\label{lemma:polySchema}  
Given the combinatorial embedding of a constant genus
graph we can find a polygonal schema for the graph in logspace.
\end{lemma}

\subsubsection{Normalizing a Polygonal Schema}

We adapt the algorithmic proof of Brahana-Dehn-Heegaard (BDH) \cite{Brahana22,DehnHeegaard07} classification
theorem as described in Vegter-Yap \cite{VegterYap90} so that it runs in logspace for constant genus graphs. The algorithm
starts with a polygonal schema and uses the following five transforms
$O(m)$ times to yield a normalized polygonal schema, where the original 
polygonal schema has $2m$ sides.
\begin{enumerate}
\item[A.] Replace $X\sigma\bar{\sigma}$ by $X$ (Example given in Figure \ref{fig:redA}).
\vspace{-2mm}
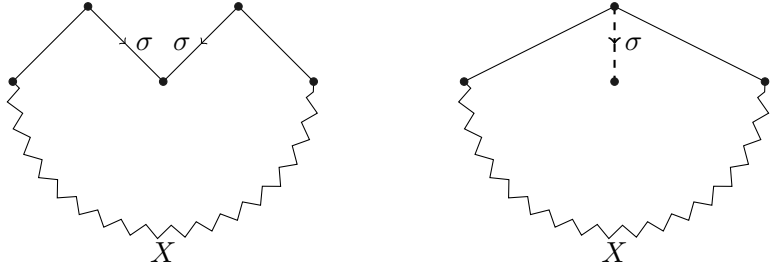
\begin{figure}[ht]
\centering
\begin{tikzpicture}
[decoration=zigzag]
\draw (0,0) -- (1,1);
\draw [->] (1,1)  -- (1.5,.5);
\draw (1.5,.5) -- (2,0);
\draw (2,0)  -- (2.5,.5);
\draw [<-] (2.5,.5) -- (3,1);
\draw (3,1) -- (4,0);
\draw [decorate] (0,0) arc(180:360:2);

\draw [fill, color=black!90] (0,0) circle(0.5 mm);
\draw [fill, color=black!90] (1,1) circle(0.5 mm);
\draw [fill, color=black!90] (2,0) circle(0.5 mm);
\draw [fill, color=black!90] (3,1) circle(0.5 mm);
\draw [fill, color=black!90] (4,0) circle(0.5 mm);

\node [right] at (1.5,.5) {$\sigma$};
\node [left] at (2.5,.5) {$\sigma$};
\node [below] at (2,-2){$X$};

\draw (6,0) -- (8,1) -- (10,0);
\draw [dashed,->,thick] (8,1) -- (8,.5);
\draw [dashed,thick] (8,.5) -- (8,0);
\draw [decorate] (6,0) arc(180:360:2);

\draw [fill, color=black!90] (6,0) circle(0.5 mm);
\draw [fill, color=black!90] (8,1) circle(0.5 mm);
\draw [fill, color=black!90] (8,0) circle(0.5 mm);
\draw [fill, color=black!90] (10,0) circle(0.5 mm);

\node [right] at (8,.5) {$\sigma$};
\node [below] at (8,-2){$X$};

\end{tikzpicture}
\caption{Reduction A (pasting along $\sigma$)}
\label{fig:redA}
\end{figure}
\item[B.] 
Replace $\sigma\tau X \bar{\tau}Y$ by 
$\rho X \bar{\rho}\sigma Y$ (Example given in Figure \ref{fig:redB}).
\vspace{-2mm}
\begin{figure}[ht]
\centering
\begin{tikzpicture}
[decoration=zigzag]
\draw [thick,->] (0,0) -- (1,0);
\draw [thick] (1,0) -- (2,0);
\draw [thick,->,dashed] (0,2) -- (1,2);
\draw [thick,dashed] (1,2) -- (2,2);
\draw [thick,->] (1,3) -- (1.5,2.5);
\draw [thick] (1.5,2.5) -- (2,2);
\draw [thick,->] (0,2) -- (0.5,2.5);
\draw [thick] (.5,2.5) -- (1,3);
\draw [decorate] (0,0) --(0,2);
\draw [decorate] (2,2) --(2,0);

\draw [fill, color=black!90] (0,0) circle(0.5 mm);
\draw [fill, color=black!90] (2,0) circle(0.5 mm);
\draw [fill, color=black!90] (0,2) circle(0.5 mm);
\draw [fill, color=black!90] (2,2) circle(0.5 mm);
\draw [fill, color=black!90] (1,3) circle(0.5 mm);

\node [left] at (.5,2.5) {$\sigma$};
\node [right] at (1.5,2.5) {$\tau$};
\node [left] at (0,1) {$Y$};
\node [right] at (2,1) {$X$};
\node [below] at (1,0) {$\tau$};
\node [below] at (1,2) {$\rho$}; 
\node[left] at (0,2) {$p$};
\node[right] at (2,2) {$q$};
\node[above] at (1,3) {$r$};
\node[right] at (2,0) {$q$};
\node[left] at (0,0) {$r$};


\draw [thick,->,dashed] (5,1) -- (6,1);
\draw [thick,dashed] (6,1) -- (7,1);
\draw [thick,->] (5,3) -- (6,3);
\draw [thick] (6,3) -- (7,3);
\draw [thick,->] (6,0) -- (6.5,.5);
\draw [thick] (6.5,.5) -- (7,1);
\draw [thick,->] (6,0) -- (5.5,.5);
\draw [thick] (5.5,.5) -- (5,1);
\draw [decorate] (5,1) --(5,3);
\draw [decorate] (7,3) --(7,1);

\draw [fill, color=black!90] (5,1) circle(0.5 mm);
\draw [fill, color=black!90] (7,1) circle(0.5 mm);
\draw [fill, color=black!90] (5,3) circle(0.5 mm);
\draw [fill, color=black!90] (7,3) circle(0.5 mm);
\draw [fill, color=black!90] (6,0) circle(0.5 mm);

\node [left] at (5.5,.5) {$\sigma$};
\node [right] at (6.5,.5) {$\rho$};
\node [left] at (5,2) {$Y$};
\node [right] at (7,2) {$X$};
\node [below] at (6,1) {$\tau$};
\node [below] at (6,3) {$\rho$}; 
\node[left] at (5,3) {$p$};
\node[right] at (7,3) {$q$};
\node[below] at (6,0) {$p$};
\node[right] at (7,1) {$q$};
\node[left] at (5,1) {$r$};

\end{tikzpicture}
\caption{Reduction B (Cutting along $\rho$ followed by pasting along $\tau$). Note that the number of vertices in the equivalence class of $r$, reduces by $1$.}
\label{fig:redB}
\end{figure}
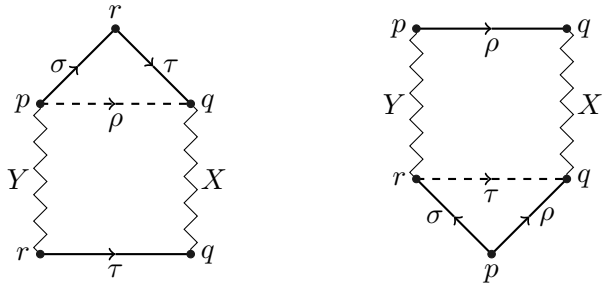
\item[C.] Replace $\sigma X \sigma Y$ by $\tau\tau Y^{*}X$, where 
$Y^{*}$ is reverse complement of $Y$ (Example given in Figure \ref{fig:redC}).
\vspace{-2mm}
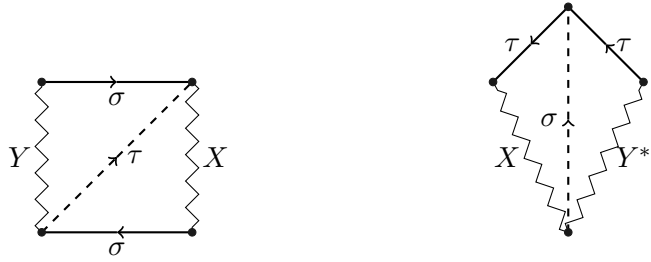
\begin{figure}[ht]
\centering
\begin{tikzpicture}
[decoration=zigzag]
\draw [thick,->] (2,0) -- (1,0);
\draw [thick] (1,0) -- (0,0);
\draw [thick,->] (0,2) -- (1,2);
\draw [thick] (1,2) -- (2,2);
\draw [thick,->,dashed] (0,0) -- (1,1);
\draw [thick,dashed] (1,1) -- (2,2);
\draw [decorate] (0,0) -- (0,2);
\draw [decorate] (2,2) -- (2,0);

\draw [fill, color=black!90] (0,0) circle(0.5 mm);
\draw [fill, color=black!90] (2,0) circle(0.5 mm);
\draw [fill, color=black!90] (0,2) circle(0.5 mm);
\draw [fill, color=black!90] (2,2) circle(0.5 mm);

\node [left] at (0,1) {$Y$};
\node [right] at (2,1) {$X$};
\node [below] at (1,0) {$\sigma$};
\node [below] at (1,2) {$\sigma$}; 
\node [right] at (1,1) {$\tau$}; 


\draw [thick,->] (8,2) -- (7.5,2.5);
\draw [thick] (7.5,2.5) -- (7,3);
\draw [thick,->] (7,3) -- (6.5,2.5);
\draw [thick] (6.5,2.5) -- (6,2);
\draw [thick,->,dashed] (7,0) -- (7,1.5);
\draw [thick,dashed] (7,1.5) -- (7,3);
\draw [decorate] (7,0) -- (8,2);
\draw [decorate] (7,0) -- (6,2);

\draw [fill, color=black!90] (7,0) circle(0.5 mm);
\draw [fill, color=black!90] (7,3) circle(0.5 mm);
\draw [fill, color=black!90] (6,2) circle(0.5 mm);
\draw [fill, color=black!90] (8,2) circle(0.5 mm);

\node [left] at (6.5,1) {$X$};
\node [right] at (7.5,1) {$Y^*$};
\node [left] at (7,1.5) {$\sigma$};
\node [left] at (6.5,2.5) {$\tau$};
\node [right] at (7.5,2.5) {$\tau$};

\end{tikzpicture}
\caption{Reduction C (Cutting along $\tau$ followed by pasting along $\rho$)}
\label{fig:redC}
\end{figure}
\item[D.] Replace $\sigma X \tau Y \bar{\sigma} U\bar{\tau}V$ by
$\rho \pi \bar{\rho}\bar{\pi} UYXV$ (Example given in Figure \ref{fig:redD}).
\vspace{-2mm}
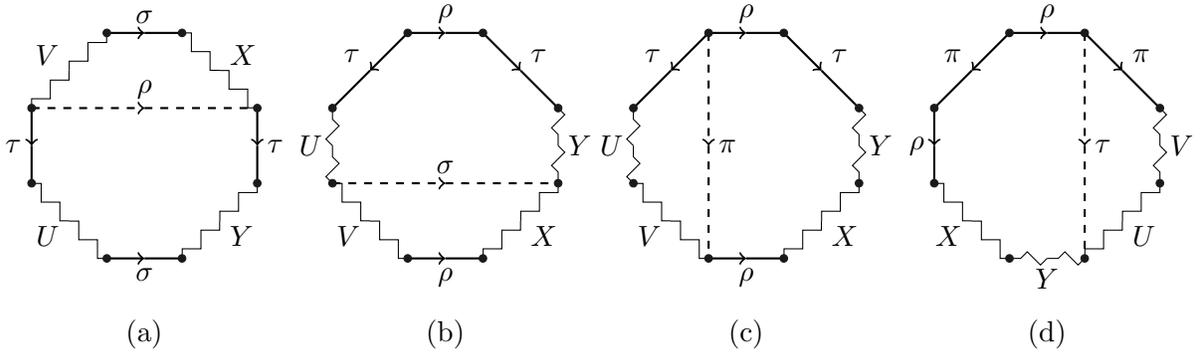
\begin{figure}[ht]
\centering
\begin{tikzpicture}
[decoration=zigzag]
\foreach \x in {0}
{
\draw [thick,->] (\x+1,3) -- (\x+1.5,3);
\draw [thick,-] (\x+1.5,3) -- (\x+2,3);
\draw [decorate] (\x+2,3) -- (\x+2.5,2.5);
\draw [decorate] (\x+2.5,2.5) -- (\x+3,2);
\draw [thick,->] (\x+3,2) -- (\x+3,1.5);
\draw [thick,-] (\x+3,1.5) -- (\x+3,1);
\draw [decorate] (\x+3,1) -- (\x+2.5,.5);
\draw [decorate] (\x+2.5,.5) -- (\x+2,0);
\draw [thick,-] (\x+2,0) -- (\x+1.5,0);
\draw [thick,<-] (\x+1.5,0) -- (\x+1,0);
\draw [decorate] (\x+1,0) -- (\x+.5,.5);
\draw [decorate] (\x+.5,.5) -- (\x,1);
\draw [thick,-] (\x,1) -- (\x,1.5);
\draw [thick,<-] (\x,1.5) -- (\x,2);
\draw [decorate] (\x,2) -- (\x+.5,2.5);
\draw [decorate] (\x+.5,2.5) -- (\x+1,3);

\draw [thick,->,dashed] (\x,2) -- (\x+1.5,2);
\draw [thick,-,dashed] (\x+1.5,2) -- (\x+3,2);

\draw [fill, color=black!90] (\x+1,3) circle(0.5 mm);
\draw [fill, color=black!90] (\x+2,3) circle(0.5 mm);
\draw [fill, color=black!90] (\x+3,2) circle(0.5 mm);
\draw [fill, color=black!90] (\x+3,1) circle(0.5 mm);
\draw [fill, color=black!90] (\x+2,0) circle(0.5 mm);
\draw [fill, color=black!90] (\x+1,0) circle(0.5 mm);
\draw [fill, color=black!90] (\x,1) circle(0.5 mm);
\draw [fill, color=black!90] (\x,2) circle(0.5 mm);

\node [above] at (\x+1.5,3) {$\sigma$};
\node [right] at (\x+2.5,2.7) {$X$};
\node [right] at (\x+3,1.5) {$\tau$};
\node [right] at (\x+2.5,.3) {$Y$};
\node [below] at (\x+1.5,0) {$\sigma$};
\node [left] at (\x+.5,.3) {$U$};
\node [left] at (\x,1.5) {$\tau$};
\node [left] at (\x+.5,2.7) {$V$};

\node [above] at (\x+1.5,2) {$\rho$};
\node  at (\x+1.5,-1) {(a)};

}

\foreach \x in {4}
{
\draw [thick,->] (\x+1,3) -- (\x+1.5,3);
\draw [thick,-] (\x+1.5,3) -- (\x+2,3);
\draw [thick,->] (\x+2,3) -- (\x+2.5,2.5);
\draw [thick,-] (\x+2.5,2.5) -- (\x+3,2);
\draw [decorate] (\x+3,2) -- (\x+3,1.5);
\draw [decorate] (\x+3,1.5) -- (\x+3,1);
\draw [decorate] (\x+3,1) -- (\x+2.5,.5);
\draw [decorate] (\x+2.5,.5) -- (\x+2,0);
\draw [thick,-] (\x+2,0) -- (\x+1.5,0);
\draw [thick,<-] (\x+1.5,0) -- (\x+1,0);
\draw [decorate] (\x+1,0) -- (\x+.5,.5);
\draw [decorate] (\x+.5,.5) -- (\x,1);
\draw [decorate] (\x,1) -- (\x,1.5);
\draw [decorate] (\x,1.5) -- (\x,2);
\draw [thick,-] (\x,2) -- (\x+.5,2.5);
\draw [thick,<-] (\x+.5,2.5) -- (\x+1,3);

\draw [thick,->,dashed] (\x,1) -- (\x+1.5,1);
\draw [thick,-,dashed] (\x+1.5,1) -- (\x+3,1);

\draw [fill, color=black!90] (\x+1,3) circle(0.5 mm);
\draw [fill, color=black!90] (\x+2,3) circle(0.5 mm);
\draw [fill, color=black!90] (\x+3,2) circle(0.5 mm);
\draw [fill, color=black!90] (\x+3,1) circle(0.5 mm);
\draw [fill, color=black!90] (\x+2,0) circle(0.5 mm);
\draw [fill, color=black!90] (\x+1,0) circle(0.5 mm);
\draw [fill, color=black!90] (\x,1) circle(0.5 mm);
\draw [fill, color=black!90] (\x,2) circle(0.5 mm);

\node [above] at (\x+1.5,3) {$\rho$};
\node [right] at (\x+2.5,2.7) {$\tau$};
\node [right] at (\x+3,1.5) {$Y$};
\node [right] at (\x+2.5,.3) {$X$};
\node [below] at (\x+1.5,0) {$\rho$};
\node [left] at (\x+.5,.3) {$V$};
\node [left] at (\x,1.5) {$U$};
\node [left] at (\x+.5,2.7) {$\tau$};

\node [above] at (\x+1.5,1) {$\sigma$};
\node  at (\x+1.5,-1) {(b)};
}

\foreach \x in {8}
{
\draw [thick,->] (\x+1,3) -- (\x+1.5,3);
\draw [thick,-] (\x+1.5,3) -- (\x+2,3);
\draw [thick,->] (\x+2,3) -- (\x+2.5,2.5);
\draw [thick,-] (\x+2.5,2.5) -- (\x+3,2);
\draw [decorate] (\x+3,2) -- (\x+3,1.5);
\draw [decorate] (\x+3,1.5) -- (\x+3,1);
\draw [decorate] (\x+3,1) -- (\x+2.5,.5);
\draw [decorate] (\x+2.5,.5) -- (\x+2,0);
\draw [thick,-] (\x+2,0) -- (\x+1.5,0);
\draw [thick,<-] (\x+1.5,0) -- (\x+1,0);
\draw [decorate] (\x+1,0) -- (\x+.5,.5);
\draw [decorate] (\x+.5,.5) -- (\x,1);
\draw [decorate] (\x,1) -- (\x,1.5);
\draw [decorate] (\x,1.5) -- (\x,2);
\draw [thick,-] (\x,2) -- (\x+.5,2.5);
\draw [thick,<-] (\x+.5,2.5) -- (\x+1,3);

\draw [thick,->,dashed] (\x+1,3) -- (\x+1,1.5);
\draw [thick,-,dashed] (\x+1,1.5) -- (\x+1,0);

\draw [fill, color=black!90] (\x+1,3) circle(0.5 mm);
\draw [fill, color=black!90] (\x+2,3) circle(0.5 mm);
\draw [fill, color=black!90] (\x+3,2) circle(0.5 mm);
\draw [fill, color=black!90] (\x+3,1) circle(0.5 mm);
\draw [fill, color=black!90] (\x+2,0) circle(0.5 mm);
\draw [fill, color=black!90] (\x+1,0) circle(0.5 mm);
\draw [fill, color=black!90] (\x,1) circle(0.5 mm);
\draw [fill, color=black!90] (\x,2) circle(0.5 mm);

\node [above] at (\x+1.5,3) {$\rho$};
\node [right] at (\x+2.5,2.7) {$\tau$};
\node [right] at (\x+3,1.5) {$Y$};
\node [right] at (\x+2.5,.3) {$X$};
\node [below] at (\x+1.5,0) {$\rho$};
\node [left] at (\x+.5,.3) {$V$};
\node [left] at (\x,1.5) {$U$};
\node [left] at (\x+.5,2.7) {$\tau$};

\node [right] at (\x+1,1.5) {$\pi$};
\node  at (\x+1.5,-1) {(c)};
}

\foreach \x in {12}
{
\draw [thick,->] (\x+1,3) -- (\x+1.5,3);
\draw [thick,-] (\x+1.5,3) -- (\x+2,3);
\draw [thick,->] (\x+2,3) -- (\x+2.5,2.5);
\draw [thick,-] (\x+2.5,2.5) -- (\x+3,2);
\draw [decorate] (\x+3,2) -- (\x+3,1.5);
\draw [decorate] (\x+3,1.5) -- (\x+3,1);
\draw [decorate] (\x+3,1) -- (\x+2.5,.5);
\draw [decorate] (\x+2.5,.5) -- (\x+2,0);
\draw [decorate] (\x+2,0) -- (\x+1.5,0);
\draw [decorate] (\x+1.5,0) -- (\x+1,0);
\draw [decorate] (\x+1,0) -- (\x+.5,.5);
\draw [decorate] (\x+.5,.5) -- (\x,1);
\draw [thick,-] (\x,1) -- (\x,1.5);
\draw [thick,<-] (\x,1.5) -- (\x,2);
\draw [thick,-] (\x,2) -- (\x+.5,2.5);
\draw [thick,<-] (\x+.5,2.5) -- (\x+1,3);

\draw [thick,->,dashed] (\x+2,3) -- (\x+2,1.5);
\draw [thick,-,dashed] (\x+2,1.5) -- (\x+2,0);

\draw [fill, color=black!90] (\x+1,3) circle(0.5 mm);
\draw [fill, color=black!90] (\x+2,3) circle(0.5 mm);
\draw [fill, color=black!90] (\x+3,2) circle(0.5 mm);
\draw [fill, color=black!90] (\x+3,1) circle(0.5 mm);
\draw [fill, color=black!90] (\x+2,0) circle(0.5 mm);
\draw [fill, color=black!90] (\x+1,0) circle(0.5 mm);
\draw [fill, color=black!90] (\x,1) circle(0.5 mm);
\draw [fill, color=black!90] (\x,2) circle(0.5 mm);

\node [above] at (\x+1.5,3) {$\rho$};
\node [right] at (\x+2.5,2.7) {$\pi$};
\node [right] at (\x+3,1.5) {$V$};
\node [right] at (\x+2.5,.3) {$U$};
\node [below] at (\x+1.5,0) {$Y$};
\node [left] at (\x+.5,.3) {$X$};
\node [left] at (\x,1.5) {$\rho$};
\node [left] at (\x+.5,2.7) {$\pi$};

\node [right] at (\x+2,1.5) {$\tau$};
\node  at (\x+1.5,-1) {(d)};
}

\end{tikzpicture}
\caption{Reduction D (a) Cutting along $\rho$. (b) Pasting along $\sigma$. (c) Cutting along $\pi$. (d) Pasting along $\tau$. }
\label{fig:redD}
\end{figure}
\item[E.] Replace 
$\sigma_1 \sigma_1 X \sigma_2 \sigma_3 \bar{\sigma_2}\bar{\sigma_3} Y$ by $\tau_1 \tau_1 \tau_2 \tau_2 \tau_3 \tau_3 XY$ (Example given in Figure \ref{fig:redE}).
\vspace{-2mm}
\begin{figure}[ht]
\centering
\begin{tikzpicture}
[decoration=zigzag]
\foreach \x in {0}
{
\draw [thick,->] (\x+1,3) -- (\x+1.5,3);
\draw [thick,-] (\x+1.5,3) -- (\x+2,3);
\draw [thick, ->] (\x+2,3) -- (\x+2.5,2.5);
\draw [thick,-] (\x+2.5,2.5) -- (\x+3,2);
\draw [decorate] (\x+3,2) -- (\x+3,1.5);
\draw [decorate] (\x+3,1.5) -- (\x+3,1);
\draw [thick,->] (\x+3,1) -- (\x+2.5,.5);
\draw [thick,-] (\x+2.5,.5) -- (\x+2,0);
\draw [thick,->] (\x+2,0) -- (\x+1.5,0);
\draw [thick,-] (\x+1.5,0) -- (\x+1,0);
\draw [thick,->] (\x+1,0) -- (\x+.5,.5);
\draw [thick,-] (\x+.5,.5) -- (\x,1);
\draw [thick,->] (\x,1) -- (\x,1.5);
\draw [thick,-] (\x,1.5) -- (\x,2);
\draw [decorate] (\x,2) -- (\x+.5,2.5);
\draw [decorate] (\x+.5,2.5) -- (\x+1,3);

\draw [thick,->,dashed] (\x+1,0) -- (\x+1.5,1.5);
\draw [thick,-,dashed] (\x+1.5,1.5) -- (\x+2,3);

\draw [fill, color=black!90] (\x+1,3) circle(0.5 mm);
\draw [fill, color=black!90] (\x+2,3) circle(0.5 mm);
\draw [fill, color=black!90] (\x+3,2) circle(0.5 mm);
\draw [fill, color=black!90] (\x+3,1) circle(0.5 mm);
\draw [fill, color=black!90] (\x+2,0) circle(0.5 mm);
\draw [fill, color=black!90] (\x+1,0) circle(0.5 mm);
\draw [fill, color=black!90] (\x,1) circle(0.5 mm);
\draw [fill, color=black!90] (\x,2) circle(0.5 mm);

\node [above] at (\x+1.5,3) {$\sigma_1$};
\node [right] at (\x+2.5,2.7) {$\sigma_2$};
\node [right] at (\x+3,1.5) {$X$};
\node [right] at (\x+2.5,.3) {$\sigma_2$};
\node [below] at (\x+1.5,0) {$\sigma_2$};
\node [left] at (\x+.5,.3) {$\sigma_3$};
\node [left] at (\x,1.5) {$\sigma_3$};
\node [left] at (\x+.5,2.7) {$Y$};

\node [above] at (\x+1.3,1.5) {$\rho$};
\node  at (\x+1.5,-1) {(a)};

}

\foreach \x in {5}
{
\draw [thick,->] (\x+1,3) -- (\x+1.5,3);
\draw [thick,-] (\x+1.5,3) -- (\x+2,3);
\draw [thick, ->] (\x+2,3) -- (\x+2.5,2.5);
\draw [thick,-] (\x+2.5,2.5) -- (\x+3,2);
\draw [thick,->] (\x+3,2) -- (\x+3,1.5);
\draw [thick,-] (\x+3,1.5) -- (\x+3,1);
\draw [decorate] (\x+3,1) -- (\x+2.5,.5);
\draw [decorate] (\x+2.5,.5) -- (\x+2,0);
\draw [thick,->] (\x+2,0) -- (\x+1.5,0);
\draw [thick,-] (\x+1.5,0) -- (\x+1,0);
\draw [thick,->] (\x+1,0) -- (\x+.5,.5);
\draw [thick,-] (\x+.5,.5) -- (\x,1);
\draw [thick,->] (\x,1) -- (\x,1.5);
\draw [thick,-] (\x,1.5) -- (\x,2);
\draw [decorate] (\x,2) -- (\x+.5,2.5);
\draw [decorate] (\x+.5,2.5) -- (\x+1,3);

\draw [thick,->,dashed] (\x+3,1) -- (\x+1.5,1.5);
\draw [thick,-,dashed] (\x+1.5,1.5) -- (\x,2);

\draw [fill, color=black!90] (\x+1,3) circle(0.5 mm);
\draw [fill, color=black!90] (\x+2,3) circle(0.5 mm);
\draw [fill, color=black!90] (\x+3,2) circle(0.5 mm);
\draw [fill, color=black!90] (\x+3,1) circle(0.5 mm);
\draw [fill, color=black!90] (\x+2,0) circle(0.5 mm);
\draw [fill, color=black!90] (\x+1,0) circle(0.5 mm);
\draw [fill, color=black!90] (\x,1) circle(0.5 mm);
\draw [fill, color=black!90] (\x,2) circle(0.5 mm);

\node [above] at (\x+1.5,3) {$\sigma_2$};
\node [right] at (\x+2.5,2.7) {$\sigma_3$};
\node [right] at (\x+3,1.5) {$\rho$};
\node [right] at (\x+2.5,.3) {$Y^*$};
\node [below] at (\x+1.5,0) {$\sigma_3$};
\node [left] at (\x+.5,.3) {$\sigma_2$};
\node [left] at (\x,1.5) {$\rho$};
\node [left] at (\x+.5,2.7) {$X$};

\node [above] at (\x+1.5,1.5) {$\sigma_1$};
\node  at (\x+1.5,-1) {(b)};

}

\foreach \x in {10}
{
\draw [thick,->] (\x+1,3) -- (\x+1.5,3);
\draw [thick,-] (\x+1.5,3) -- (\x+2,3);
\draw [thick, ->] (\x+2,3) -- (\x+2.5,2.5);
\draw [thick,-] (\x+2.5,2.5) -- (\x+3,2);
\draw [thick,->] (\x+3,2) -- (\x+3,1.5);
\draw [thick,-] (\x+3,1.5) -- (\x+3,1);
\draw [thick,->] (\x+3,1) -- (\x+2.5,.5);
\draw [thick,-] (\x+2.5,.5) -- (\x+2,0);
\draw [thick,->] (\x+2,0) -- (\x+1.5,0);
\draw [thick,-] (\x+1.5,0) -- (\x+1,0);
\draw [thick,->] (\x+1,0) -- (\x+.5,.5);
\draw [thick,-] (\x+.5,.5) -- (\x,1);
\draw [decorate] (\x,1) -- (\x,1.5);
\draw [decorate] (\x,1.5) -- (\x,2);
\draw [decorate] (\x,2) -- (\x+.5,2.5);
\draw [decorate] (\x+.5,2.5) -- (\x+1,3);

\draw [fill, color=black!90] (\x+1,3) circle(0.5 mm);
\draw [fill, color=black!90] (\x+2,3) circle(0.5 mm);
\draw [fill, color=black!90] (\x+3,2) circle(0.5 mm);
\draw [fill, color=black!90] (\x+3,1) circle(0.5 mm);
\draw [fill, color=black!90] (\x+2,0) circle(0.5 mm);
\draw [fill, color=black!90] (\x+1,0) circle(0.5 mm);
\draw [fill, color=black!90] (\x,1) circle(0.5 mm);
\draw [fill, color=black!90] (\x,2) circle(0.5 mm);

\node [above] at (\x+1.5,3) {$\tau_1$};
\node [right] at (\x+2.5,2.7) {$\tau_1$};
\node [right] at (\x+3,1.5) {$\tau_2$};
\node [right] at (\x+2.5,.3) {$\tau_2$};
\node [below] at (\x+1.5,0) {$\tau_3$};
\node [left] at (\x+.5,.3) {$\tau_3$};
\node [left] at (\x,1.5) {$X$};
\node [left] at (\x+.5,2.7) {$Y$};

\node  at (\x+1.5,-1) {(c)};

}

\end{tikzpicture}
\caption{Reduction E (a) Cutting along $\rho$. (b) Pasting along $\sigma_1$. (c) Obtained from Figure \ref{fig:redE}(b) by applying Reduction C thrice.}
\label{fig:redE}
\end{figure}
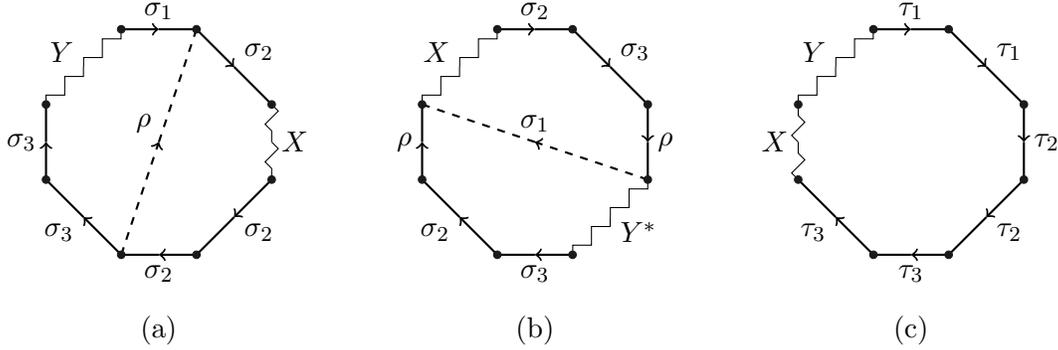
\item[F.] Replace $\sigma \sigma \tau \tau X$
by $\sigma \rho \bar{\sigma}  \rho X$ (Example given in Figure \ref{fig:redF}).
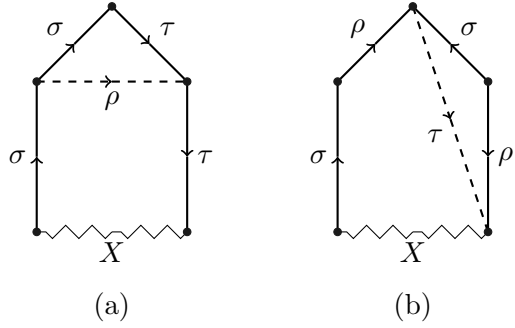
\begin{figure}[ht]
\centering
\begin{tikzpicture}
[decoration=zigzag]
\foreach \x in {0}
{
\draw [thick,->] (\x,0) -- (\x,1);
\draw [thick,-] (\x,1) -- (\x,2);
\draw [thick, ->] (\x,2) -- (\x+.5,2.5);
\draw [thick,-] (\x+.5,2.5) -- (\x+1,3);
\draw [thick,->] (\x+1,3) -- (\x+1.5,2.5);
\draw [thick,-] (\x+1.5,2.5) -- (\x+2,2);
\draw [thick,->] (\x+2,2) -- (\x+2,1);
\draw [thick,-] (\x+2,1) -- (\x+2,0);
\draw [decorate] (\x+2,0) -- (\x+1,0);
\draw [decorate] (\x+1,0) -- (\x,0);

\draw [thick,->,dashed] (\x,2) -- (\x+1,2);
\draw [thick,-,dashed] (\x+1,2) -- (\x+2,2);

\draw [fill, color=black!90] (\x,0) circle(0.5 mm);
\draw [fill, color=black!90] (\x,2) circle(0.5 mm);
\draw [fill, color=black!90] (\x+1,3) circle(0.5 mm);
\draw [fill, color=black!90] (\x+2,2) circle(0.5 mm);
\draw [fill, color=black!90] (\x+2,0) circle(0.5 mm);

\node [left] at (\x,1) {$\sigma$};
\node [left] at (\x+.5,2.7) {$\sigma$};
\node [right] at (\x+1.5,2.7) {$\tau$};
\node [right] at (\x+2,1) {$\tau$};
\node [below] at (\x+1,0) {$X$};

\node [below] at (\x+1,2) {$\rho$};
\node  at (\x+1,-1) {(a)};
}

\foreach \x in {4}
{
\draw [thick,->] (\x,0) -- (\x,1);
\draw [thick,-] (\x,1) -- (\x,2);
\draw [thick, ->] (\x,2) -- (\x+.5,2.5);
\draw [thick,-] (\x+.5,2.5) -- (\x+1,3);
\draw [thick,-] (\x+1,3) -- (\x+1.5,2.5);
\draw [thick,<-] (\x+1.5,2.5) -- (\x+2,2);
\draw [thick,->] (\x+2,2) -- (\x+2,1);
\draw [thick,-] (\x+2,1) -- (\x+2,0);
\draw [decorate] (\x+2,0) -- (\x+1,0);
\draw [decorate] (\x+1,0) -- (\x,0);

\draw [thick,->,dashed] (\x+1,3) -- (\x+1.5,1.5);
\draw [thick,-,dashed] (\x+1.5,1.5) -- (\x+2,0);

\draw [fill, color=black!90] (\x,0) circle(0.5 mm);
\draw [fill, color=black!90] (\x,2) circle(0.5 mm);
\draw [fill, color=black!90] (\x+1,3) circle(0.5 mm);
\draw [fill, color=black!90] (\x+2,2) circle(0.5 mm);
\draw [fill, color=black!90] (\x+2,0) circle(0.5 mm);

\node [left] at (\x,1) {$\sigma$};
\node [left] at (\x+.5,2.7) {$\rho$};
\node [right] at (\x+1.5,2.7) {$\sigma$};
\node [right] at (\x+2,1) {$\rho$};
\node [below] at (\x+1,0) {$X$};

\node [below] at (\x+1.3,1.5) {$\tau$};
\node  at (\x+1,-1) {(b)};
}

\end{tikzpicture}
\caption{Reduction F (a) Cutting along $\rho$. (b) Pasting along $\tau$.}
\label{fig:redF}
\end{figure}
\end{enumerate}
The procedure is to 
\begin{enumerate}
\item Use reductions A,B,C several times to ensure that all the 
sides of the polygonal schema have a common endpoint. 
\item 
\begin{enumerate}
\item Orientable case: Use transform D repeatedly to bring the polygon 
in normal form. 
\item Non-orientable case:
\begin{itemize}
\item[-] Use reductions C,D to convert the schema
into a form where the orientable symbols are clustered and non-orientable symbols are paired
\item[-] Use reduction E repeatedly (in the forward direction) to eliminate all 
orientable symbols.
\item[-] Use Reduction E in the \emph{reverse} direction repeatedly 
to eliminate all but at most one non-orientable symbol.
\item[-] Use Reduction F, if necessary, to ensure that there is at most one 
non-orientable symbol.
\end{itemize}
\end{enumerate}
\end{enumerate}
Possibly, the only step requiring any explanation is the last one. We apply
Reduction E in reverse with $X$ as the empty string to replace three
non-orientable symbols by two orientable ones forming a cluster of $4$ and
a single non-orientable one which forms a pair. The way we apply the reduction,
ensures that both the orientable and the non-orientable parts are contiguous.

Fianlly we will be left with a string in one of the first two normal forms 
or a string of the form $\sigma \sigma \tau \tau X$ (where $X$ is an orientable
schema in normal form) in which case Reduction F is applicable. 

To see that the above procedure can be carried out in \L\ it suffices to prove
that each of the above reductions can be carried out in \L, the number of 
reductions is bounded by a constant and we can decide in \L\ when to carry out 
a reduction.

The Vegter-Yap paper does careful book-keeping in order to ensure that the
number of operations in Step 1 is linear in the original genus. 
We can alternatively, follow the brute force approach and keep on applying
Reductions A,B,C while the sides of the polygon do not have a common
end-point. This will require at most linear number of applications of 
the first two reductions.

Observe that for the orientable case, each application of reduction D
reduces the number of unclustered symbols by two. Thus we are done in 
$O(m)$ applications of this reduction. Similarly, each application of
reduction C reduces the number of unpaired non-orientable symbols by one
and as before every application of reduction D reduces the number of 
unclustered orientable symbols by two. So in $O(m)$ steps all the
orientable symbols are clustered and the non-orientable symbols are
paired. Now every application of reduction E in the forward direction
gets rid of two orientable symbols so in $O(m)$ steps all the orientable
symbols are removed. Finally $O(m)$ applications of reduction E in reverse
lead to removal of all but one non-orientable symbols.

To see that each of the steps is in \L\ observe that each of the steps
involves one or more of the following operations:
\begin{itemize}
\item[-] find a path through the interior of the polygon between two points 
on its boundary 
\item[-] cut along a path
\item[-] paste two paired sides of (a cut) polygon together
\end{itemize}
We know how to do the second operation in \L\ while the third, being the 
reverse of the second one is even easier, since we just have to identify
corresponding spurious vertices and then excise them out of the corresponding
edge. The first operation is just an undirected reachability question in the 
graph (minus its boundary) hence is in \L\ by Reingold's Theorem. 

Finally, a determination of when to apply a particular reduction is easily
seen to be in \L\ for all but, possibly, reduction D. In this case,
for an orientable symbol $\sigma$ separated from its mate $\bar{\sigma}$
on both sides, sequentially test for each other symbol $\tau$ 
if it lies in one
of the two stretches that $\sigma$ and its mate divide the schema into, while
its mate $\bar{\tau}$ lies in the other. Having found the first such $\tau$
suffices to enable a use of the reduction.

Thus, using the above argument and Lemma~\ref{lemma:polySchema} 
we have sketched the proof of the following theorem:
\begin{theorem} \label{theorem:normal}
Given a combinatorial embedding of constant genus, say $g$ (which is 
positive or otherwise), for a graph $G$,
in logspace we can find a polygonal schema for the graph in normal form. 
of genus $O(|g|)$ in magnitude, and also the corresponding combinatorial
embedding.
\end{theorem}

Let {\kGonBi} be the class of constant genus, bipartite graphs along with an embedding given on the polygonal schema in normal form of the surface in which the graph has an embedding. Moreover, for every graph in this class, no edge has both its end points on the boundary of the polygon.

At this point, there are no vertices lying on the boundary of the 
polygonal schema, only edges crossing it. It is easy to see that for
each such edge $e = (u,v)$ which has two halves lying on segments of the
polygon, if we introduce internal vertices $u'=v', v''$ on the edge (converting
it to a path $u,u' = v', v'', v$) so that $u',v'$ lie on the boundary of the 
polygon on the sides nearer to $u,v$ respectively, then, because the path has 
odd length the number of perfect matchings in the modified graph is preserved.

Thus we have proved that:
\begin{corollary} \label{cor:PS_reduction}
Given the combinatorial embedding of a graph of constant genus,
there is an logspace reduction, which preserves perfect matchings, to a graph
in the class {\kGonBi}.
\end{corollary}

\subsubsection{From Polygonal Schema in normal form to a Grid}

\begin{lemma} \label{lemma:grid_reduction}
If $G$ is an orientable graph in {\kGonBi}, then one can get a logspace, matching-preserving reduction form $G$ to a graph $H \in \GG$
\end{lemma}
\begin{proof}
We start with a graph $G \in \kGonBi$ and construct a graph $H \in \GG$ such that the number of perfect matchings in $G$ and $H$ are the same.
 
We can assume that the maximum degree of $G$ is $3$  and there exists a vertex $s$ of degree $2$ \cite{KMV08}. Think of $G$ as a planar graph. Reduce $G$ to a grid graph $G'$ using \cite{ABCDR09}. It follows from the reduction that faces are preserved (modulo subdivision of edges). Let $T$ be the spanning tree of $G$ constructed by the algorithm that would be embedded on the course grid and let $T'$ be the tree corresponding to $T$ in $G'$. Every vertex (say $v$) on boundary of the polygon in $G$ is a leaf node since every edge has at most one of its end points on the boundary of the polygon (by definition of {\kGonBi}). Therefore $v$ is also a leaf in $T$. Let $s$ be the root of $T$ and $h(u)$ be the height of a vertex $u$ in $T$. It follows from the reduction that $h(u)$ is the value of its $y$-coordinate in $G'$. 

For the rest of this proof we will use the notation $u'$ and $v'$ to denote the respective copies of some two vertices $u$ and $v$ in $G$. Now subdivide every horizontal edge in $G'$ into $2$ edges to get the grid graph $G''$. This ensures that the horizontal distance between the copies of any two vertices in $G'$ is even. First claim is that the number of matchings in $G$ and $G''$ are the same. To see this it is enough to show that: $e=(u,v)$ is an edge in $G$ iff any simple path from $u'$ to $v'$ has odd length. If $e$ is a tree edge then the vertical distance between $u'$ and $v'$ is $1$ and the horizontal distance is even. Thus the distance between them on the grid is odd and therefore any path between them on the grid has odd length. Similarly, if $e$ is a non-tree edge, then $h(u)$ and $h(v)$ have different parity and therefore the vertical distance between them is odd. 

Now we will see how to construct the grid graph $H$ as required by the Lemma. Let $G''$ be a $m_1 \times m_2$ grid. Construct an empty grid $H$, of size $(m_1+2) \times (m_2+2)$. Place the grid $G''$ on $H$ so that $G''$ lies properly inside the grid (that is no edge of $G''$ has an end point on any of the boundary vertices of $H$). Suppose two vertices $u$ and $v$ in $G$ get identified when $G$ is thought of as a genus $g$ graph. Then from our earlier observation we have that both $u'$ and $v'$ must be leaf nodes and lie on the outer face of $G''$. Also $h(u)$ and $h(v)$ must have the same parity, since otherwise we can construct an odd cycle in $G$ by traversing from $s$ to $u$ (which is the same as $v$) and back to $s$ via $v$. This implies that the $y$-coordinate of both $u'$ and $v'$ in $G''$ has the same parity. Drop a path from $u$ (and similarly a path from $v$) by going down all the way to the south border of $H$. Observe that  the sum of the lengths of these two paths is even. This is because, the difference in their $y$-coordinates is even. This ensures that matching is preserved by adding these paths.

The ordering of the segments in the outer face that get glued, is same in both $H$ and $G$ since faces are preserved by the reduction in \cite{ABCDR09}. Also the by our construction the length of each segment is even since the horizontal distance between two vertices is a multiple of $2$. Additionally from there are no edges along the boundary of the grid as required.
\end{proof}

\subsection{Any graph in a``genus $g$ grid" is bipartite}
 \begin{lemma}
\label{lem:even}
Any graph $G \in \GG$ is bipartite.
\end{lemma}
\begin{proof}
Let $C$ be a cycle in $G$. First we consider the case when $C$ is a simple cycle. Partition $C$ into paths $P_1 = (p_1,\ldots, p_2) ,P_2 = (p_2,\ldots, p_3), \ldots , P_k = (p_k,\ldots, p_1)$, such that each $P_i$ lies entirely in the grid with its two end points $p_i$ and $p_{i+1}$ lying on some two segments. An example of this partition is shown in Figure \ref{fig:example_cycle} for the respective cycle. For each path $P_i$ construct a path $P'_i$ by moving along the border of the grid from $p_i$ to $p_{i+1}$ along a fixed direction (say in clockwise direction). 

Fix an $i \in [k]$. Consider the partition of $P'_i$, induced by the segments along which it passes. Denote the first and the last partition by $P'_{i_1}$ and $P'_{i_2}$ respectively. Note that any of the intermediate partitions of $P'_i$ has even length since the length of an intermediate partition equals the length of the corresponding segment and hence is even. Therefore we have,
\begin{eqnarray}
\label{eqn:one}
|P'_i| \mod 2 & = & (|P'_{i_1}| + |P'_{i_2}|) \mod 2.
\end{eqnarray}
Also, 
\begin{eqnarray}
\label{eqn:two}
|P_i| \mod 2 &=& |P'_i| \mod 2, 
\end{eqnarray}
because the path $P_i$ and $P'_i$ together form a simple cycle on the grid and any cycle that lies entirely on the grid has even length. Consider the sum,
\begin{eqnarray}
\mathcal{S} &=& \sum_{i=1}^k (|P'_{i_1}| + |P'_{i_2}|).
\end{eqnarray}
Rearranging we get, 
\begin{eqnarray}
\mathcal{S} & = & |P'_{1_1}| + |P'_{k_2}| + \sum_{i=1}^{k-1} (|P'_{i_2}| + |P'_{{(i+1)}_1}|).
\end{eqnarray}
Since $|P'_{i_2}|$ and $|P'_{{((i+1) \mod k)}_1}|$ are equal, we have,
\begin{eqnarray}
\label{eqn:five}
\mathcal{S} & = & 2|P'_{1_1}|  + \sum_{i=1}^{k-1} 2|P'_{i_2}|.
\end{eqnarray}

Now combining Equations \eqref{eqn:one}, \eqref{eqn:two} and \eqref{eqn:five}, we have $ \sum_{i=1}^k |P_i|$ is even and thus $C$ is of even length.

If $C$ is non-simple, then $C$ can be decomposed into a collection of simple cycles $\{C_j\}$ such that $|C| = \sum_j C_j$. Now using the previous part we get that $C$ has even length.
\end{proof}

\section{New Upper Bounds}

In this section we establish new upper bounds on the space complexity of certain matching problems on bipartite constant genus graphs, embedded on a `genus $g$ grid'.

We define {\GG} to be the class of genus $g$ graphs such that: for every $G \in \GG$, $G$ is a grid graph embedded on a grid of size $2m \times 2m$. We assume that the distance between adjacent horizontal (and similarly vertical) vertices is of unit length. The entire boundary of the grid is divided into $4g$ segments, and each segment has even length, for some constant $g$. The $4g$ segments are labelled as $(S_1, S_2, S_1', S_2', \ldots S_{2i-1}, S_{2i}, S_{2i-1}', S_{2i}',$ $\ldots , S_{2g-1}, S_{2g}, S_{2g-1}', S_{2g}')$, together with a direction, namely, $S_i$ is directed from left to right and $S_i'$ is directed from right to left for each $i \in [2g]$. The $j$th vertex on a segment $S_i$ is the $j$th vertex on the border of the grid, starting from the head of the segment $S_i$ and going along the direction of the segment. Finally the segments $S_i$ and $S_i'$ are glued to each other for each $i \in [2g]$ in the same direction. In other words, the $j$th vertex on segment $S_i$ is the same as the $j$th vertex on segment $S_i'$. Also there are no edges along the boundary of the grid.

\begin{definition}
If $C$ is a cycle in $G$, we denote the circulation of $C$ with respect to a weight function $w$ as $circ_w(C)$. For any subset $E' \subseteq C$, $circ_w(E')$ is the value of the circulation restricted to the edges of $E'$.
\end{definition}

\begin{theorem}[Main Theorem]
\label{theorem:main}
There exists a logspace computable and polynomially bounded weight function $W$, such that for any graph $G \in \GG$ and any cycle $C \in G$, $circ_W(C) \ne 0$.
\end{theorem}

\begin{theorem}
For a graph embedded on a constant genus surface,
\begin{itemize}
\item[(a)]
{\isMatch} is in {\SPL},
\item[(b)]
{\constructMatch} is in $\FL^{\SPL}$ and 
\item[(c)]
{\isUniqueMatch} is in {\SPL}. 
\end{itemize}
\end{theorem}
\begin{proof}
As a result of Theorem \ref{theorem:main_reduction}, we can assume that our input graph $G \in \GG$. Using Theorem \ref{theorem:main} and Lemma \ref{lemma:uniquepm} we get a logspace computable weight function $W$, such that the minimum weight perfect matching in $G$ with respect to $W$ is unique. Moreover, for any subset $E' \subseteq E$, Theorem \ref{theorem:main} is valid for  the subgraph $G \setminus E'$ also, with respect to the same weight function $W$. Now (a) and (b) follows from Lemma \ref{lemma:spl}.  Checking for uniqueness can be done by first computing a perfect matching, then deleting an edge from the matching and rechecking to see if a perfect matching exists in the new graph. If it does, then $G$ did not have a unique perfect matching, else it did. Note that Theorem \ref{theorem:main} is valid for any graph formed by deletion of edges of $G$.
\end{proof}

Theorem \ref{theorem:main} also gives an alternative proof of directed graph reachability for constant genus graphs.
\begin{theorem}[\cite{BTV09,KV09}]\label{theorem:reachability}
Directed graph reachability for constant genus graphs is in {\UL}.
\end{theorem}
\noindent The proof of Theorem \ref{theorem:reachability} follows from Lemma \ref{lemma:reach} and \cite{BTV09}.
We adapt Lemma \ref{lemma:reach} from the journal version of \cite{DKR08} (to appear in Theory of Computing Systems).
\begin{lemma}
\label{lemma:reach}
There exist a logspace computable weight function
that assigns polynomially bounded weights 
to the edges of a directed graph such that:
(a) the weights are skew symmetric, i.e., 
w(u,v) = - w(v,u),  and 
(b) the sum of weights along any (simple) directed
cycle is non-zero.
\end{lemma}

\begin{lemma}
In any class of graphs closed under the subdivision of edges,
{\em Theorem \ref{theorem:main}} 
implies the hypothesis of 
{\em Lemma \ref{lemma:reach}}.
\end{lemma}
\begin{proof}
Given an undirected graph $G,$ construct a {\em bipartite}
graph $G'$ as follows:
replace every undirected edge $\{u,v\}$ by
a path $u-w-v$ of length two. 
Use Lemma \ref{lemma:reach} to assign weights to the edges of $G'.$
Suppose that the weight assigned to the undirected edge $\{u,w\}$ in $G'$ is $a$ and the weight of $\{w,v\}$ is $b.$ 
Let $\overrightarrow{G}$ denote the directed graph obtained
from $G$ by considering each undirected edge as two 
directed edges in opposite directions.
Now we assign the weights to the edges of $\overrightarrow{G}$ as follows:
directed edge $(u,v)$ gets weight $a - b;$ whereas the directed edge $(v,u)$ will get weight $b - a.$ 
The circulations of the cycles in $G'$ being non-zero will
translate into the sum of the edges along any cycle
in the directed graph $\overrightarrow{G}$ being non-zero.
\end{proof}




\subsection{Proof of Main Theorem}
\begin{proof}[Proof of Theorem \ref{theorem:main}]
For a graph $G \in \GG$, we define $W$ is a linear combination of the following $4g+1$ weight functions defined below. This is possible in logspace since $g$ is constant.


Define $4g+1$ weight functions as follows:

\begin{itemize}

\item[-]
For each $i \in [2g]$,
\begin{equation}
w_{i}(e) = 
\left\{
\begin{array}{ll}
     1 & \textrm{if $e$ lies on the segment $S_i$} \\
 0 & \textrm{otherwise} \\
\end{array}
\right.
\end{equation}

\item[-]
For each $i \in [2g]$,

\begin{equation}
w'_{i}(e) = 
\left\{
\begin{array}{ll}
     j & \textrm{if $e$ lies on the segment $S_i$ at index $j$ from the head of $S_i$ and $j$ is odd} \\
     -j & \textrm{if $e$ lies on the segment $S_i$ at index $j$ from the head of $S_i$ and $j$ is even}  \\    
 0 & \textrm{otherwise} \\
\end{array}
\right.
\end{equation}

\item[-]
\begin{equation}
w''(e) = 
\left\{
\begin{array}{ll}
     0 & \textrm{if one end of $e$ lies on the boundary of the grid} \\
     0 & \textrm{if $e$ does not lie on the boundary and $e$ is a vertical edge} \\
     (-1)^{i+j} (i+j-1) & \textrm{if $e$ is the $j$th horizontal edge from left, lying in row $i$}\\
     & \textrm{from bottom, and not lying on the boundary}\\    
\end{array}
\right.
\end{equation}

\end{itemize}

Note that if $e$ does not lie on the boundary of the grid then $w''(e)$ is same as the weight function defined in \cite{DKR08}. 

If $C$ is a cycle in $G$, we denote the circulation of $C$ with respect to a weight function $w$ as $circ_w(C)$. For any subset $E' \subseteq C$, $circ_w(E')$ is the value of the circulation restricted to the edges of $E'$. An example of a cycle on a grid is given in Figure \ref{fig:example_cycle}.

Let $C$ be a simple cycle in $G$. If $C$ is surface non-separating, then $circ_{w_i}(C) \ne 0$ for some $i$. If $C$ is surface separating and crosses the boundary of the grid at some vertex $v$, then $circ_{w'_i}(C) \ne 0$ for $i$, such that $v$ lies in the segment $S_i$. If $C$ does not intersect any of the boundary segments, then $C$ does not have any edge on the boundary since there are no edges along the boundary by definition of {\GG}. Therefore $circ_{w''}(C) \ne 0$ by \cite{DKR08}.

Without loss of generality, assume $C$ intersects segment $S_1$. Let $E^C_1$ be the set of edges of $C$ that intersect $S_1$. Note that $circ_{w_1}(C) = circ_{w_1}(E^C_1)$ (same thing holds for $w'_1$ as well. We can assume that $|E^C_1|$ is even since otherwise $circ_{w_1}(E^C_1)$ is odd and hence non-zero. By Lemma \ref{lem:alternate} it follows that the edges of $E^C_1$, alternate between going out and coming into the grid. Then using Lemma \ref{lem:weight} we get that $circ_{w'_1}(E^C_1) \ne 0$ and thus $circ_{w'_1}(C) \ne 0$. (See below for Lemma \ref{lem:alternate} and \ref{lem:weight})
\end{proof}

\begin{figure}
\centering
\subfigure[Example of a cycle on the grid that crosses each segment an even number of times with the weights $w'_1$] 
{
    \label{fig:example_cycle}
\begin{tikzpicture}[scale=.4,shorten >=.35mm,>=latex]

 \tikzstyle gridlines=[color=black!20,very thin]
 \draw[color=black!50,very thin] (0,0) grid (12,12);

 \foreach \x in {1,...,12}
  \foreach \y in {1,...,12}
   {
     \draw[fill,color=black!90] (\x,\y) circle (0.35mm);
   }

\draw [<-] (13,0) -- (13,6);
\draw [<-] (13,6) -- (13,12);
\draw [->] (-1,0) -- (-1,6);
\draw [->] (-1,6) -- (-1,12);
\draw [->] (0,-1) -- (6, -1);
\draw [->] (6,-1) -- (12,-1);
\draw [<-] (0,13) -- (6, 13);
\draw [<-] (6,13) -- (12, 13);

\node at (3,-1.5) {$S_1$};
\node at (9,-1.5) {$S_2$};
\node at (3,13.5) {$S_4$};
\node at (9,13.5) {$S_3$};
\node at (13.5,3) {$T_1$};
\node at (13.5,9) {$T_2$};
\node at (-1.5,9) {$T_3$};
\node at (-1.5,3) {$T_4$};

\draw [very thick,-] (1,0) -- (1, 2) --(2,2) -- (2,0) ;
\node [right] at (2,1.5) {$P_1$};
\draw [very thick,-] (12,5) -- (8,5) -- (8,12);
\node [left] at (8,9.5) {$P_2$};
\draw [very thick,-] (0,11) -- ( 6,11) -- (6,0);
\node [right] at (6,4.5) {$P_3$};
\draw [very thick,-] (12,1) -- (9,1) -- (9,3) -- (12,3);
\node [left] at (9,1.5) {$P_4$};
\draw [very thick,-] (4,0) -- (4,3) -- (0,3);
\node [right] at (4,2.5) {$P_5$};
\draw [very thick,-] (4,12) -- (4,12) -- (2,12) -- (2,12);
\node [left] at (4,12.5) {$P_6$};
\draw [very thick,-] (0,5) -- (3,5) -- (3,8) -- (0,8);
\node [left] at (3,6.5) {$P_7$};
\draw [very thick,-] (11,12) -- (11,6) -- (12,6);
\node [left] at (11,7.5) {$P_8$};

\node at (1,-0.5) {$1$};
\node at (2,-0.5) {$-2$};
\node at (3,-0.5) {$3$};
\node at (4,-0.5) {$-4$};
\node at (5,-0.5) {$5$};
\node at (6,-0.5) {$-6$};

\end{tikzpicture}
}
\hspace{1cm}
\subfigure[Construction of a path from $Q_1$ to $Q_2$ in $\Gamma \setminus C$ (the dotted path is the shortest path between $Q_1$ and $Q_1'$ (resp. between $Q_2$  and $Q_2'$).] 
{
    \label{fig:intersection}
\begin{tikzpicture}[scale=1,shorten >=.35mm,>=latex]

\draw [->] (0,0) -- (6,0);
\draw [->] (2,-2) -- (2,2);
\draw [->] (4,-2) -- (4,2);
\draw [-,very thick] (2.2,-1) -- (2.2,0) -- (3.8,0) -- (3.8,2);
\draw [-,very thick] (1.8,-2) -- (1.8, -1);
\draw [dashed, thick, -] (1.8,-1) -- (0,-1);
\draw [dashed, thick,-] (3,-2) -- (2.2,-1); 

\draw[fill,color=black!90] (3,-2) circle (0.5mm);
\draw[fill,color=black!90] (0,-1) circle (0.5mm);
\draw[fill,color=black!90] (2.2,0) circle (0.5mm);
\draw[fill,color=black!90] (3.8,0) circle (0.5mm);
\draw[fill,color=black!90] (2.2,-1) circle (0.5mm);
\draw[fill,color=black!90] (1.8,-1) circle (0.5mm);

 \draw [below] node at (0,-1) {$Q_2$};
 \draw [right] node at (3,-2) {$Q_1$};
 \draw [above] node at (1.6,-1) {$Q_2'$};
 \draw [right] node at (2.2,-1) {$Q_1'$};
 \draw [left] node at (2,1.5) {$C$};
 \draw [right] node at (4,1.5) {$C$};
 \draw [above] node at (2.3,0) {$P_1$};
 \draw [below] node at (3.7,0) {$P_2$};
 \draw [below] node at (6,0) {$C_j$};
\end{tikzpicture}
}
\caption{}
\label{fig:example} 
\end{figure}

To establish Lemma \ref{lem:alternate} we use an argument (Lemma \ref{lem:cabello}) from homology theory. For two cycles (directed or undirected) $C_1$ and $C_2$, let $I(C_1,C_2)$ denote the number of times $C_1$ and $C_2$ cross each other (that is one of them goes from the left to the right side of the other, or vice versa). 

Next we adapt the following Lemma from Cabello and Mohar \cite{CabelloMohar07}. Here we assume we are given an orientable surface (Cabello and Mohar gives a proof for a graph on a surface).
\begin{lemma}[\cite{CabelloMohar07}]
\label{lem:cabello}
Given a genus $g$ orientable, surface $\Gamma$, let $\mathcal{C} = \{C_i\}_{i \in[2g]}$ be a set of cycles that generate the first homology group $H_1(\Gamma)$. A cycle $C$ in $\Gamma$ in non-separating if and only if there is some cycle $C_i \in \mathcal{C}$ such that $I(C,C_i) \equiv 1 (\mod 2)$.
\end{lemma}
\begin{proof}
Let $\tilde{C}$ be some cycle in $\Gamma$. We can write $\tilde{C} = \sum_{i \in [2g]} t_i C_i$ since $\mathcal{C}$ generates $H_1(\Gamma)$. Define $I_{\tilde{C}} (C) = \sum_{i \in [2g]} t_i I(C,C_i) (\mod 2)$. One can verify that $I_{\tilde{C}} : \mathcal{C}_1({\Gamma}) \rightarrow \mathbb{Z}_2$ is a group homomorphism. Now since $\mathcal{B}_1(\Gamma)$ is a normal subgroup of $\mathcal{B}_1(\Gamma)$, $I_{\tilde{C}}$ induces a homomorphism from $H_1(\Gamma)$ to $\mathbb{Z}_2$.

Any cycle is separating if and only if it is homologous to the empty set. Therefore if $C$ is separating, then $C \in \mathcal{B}_1(\Gamma)$ and thus every homomorphism from $H_1(\Gamma)$ to $\mathbb{Z}_2$ maps it to $0$. Hence for every $i \in [2g]$, $I(C,C_i) \equiv I_{C_i} (C) = 0$.

Suppose $C$ is non-separating. One can construct a cycle $C'$ on $\Gamma$, that intersects $C$ exactly once. Let $C' =  \sum_{i \in [2g]} t_i' C_i$. Now $1 \equiv I_{C'}(C) \equiv \sum_{i \in [2g]} t_i' I(C,C_i) (\mod 2)$. This implies that there exists $i \in [2g]$ such that $I(C,C_i) \equiv 1(\mod 2)$.
\end{proof}
\begin{lemma}
\label{lem:alternate}
Let $C$ be a simple directed cycle on a genus $g$ orientable surface $\Gamma$ and let $\mathcal{C} = \{C_i\}_{i \in[2g]}$ be a system of $2g$ directed cycles on $\Gamma$, having exactly one point in common and $\Gamma \setminus \mathcal{C}$ is the fundamental polygon, say $\Gamma'$. If $I(C,C_i)$ is even for all $i \in [2g]$ then for all $j \in [2g]$, $C$ alternates between going from left to right and from right to left of the cycle $C_j$ in the direction of $C_j$ (if $C$ crosses $C_j$ at all).
\end{lemma}
\begin{proof}
Suppose there exists a $j \in [2g]$ such that $C$ does not alternate being going from left to right and from right to left with respect to $C_j$. Thus if we consider the ordered set of points where $C$ intersects $C_j$, ordered in the direction of $C_j$, there are two consecutive points (say $P_1$ and $P_2$) such that at both these points $C$ crosses $C_j$ in the same direction. 

Let $Q_1$ and $Q_2$ be two points in $\Gamma \setminus C$. We will show that there exists a path in $\Gamma \setminus C$ between $Q_1$ and $Q_2$. Consider the shortest path from $Q_1$ to $C$. Let $Q_1'$ be the point on this path that is as close to $C$ as possible, without lying on $C$. Similarly define a point $Q_2'$ corresponding to $Q_2$. Note that it is sufficient for us to construct a path between $Q_1'$ and $Q_2'$ in $\Gamma \setminus C$. If both $Q_1'$ and $Q_2'$ locally lie on the same side of $C$, then we get a path from $Q_1'$ to $Q_2'$ not intersecting $C$, by traversing along the boundary of $C$. Now suppose $Q_1'$ and $Q_2'$ lie on opposite sides (w.l.o.g. assume that $Q_1'$ lies on the right side) of $C$. From $Q_1'$ start traversing the cycle until you reach cycle $C_j$ (point $P_1$ in Figure \ref{fig:intersection}). Continue along cycle $C_j$ towards the adjacent intersection point of $C$ and $C_j$, going as close to $C$ as possible, without intersecting it (point $P_2$ in Figure \ref{fig:intersection}). Essentially this corresponds to switching from one side of $C$ to the other side without intersecting it. Next traverse along $C$ to reach $Q_2'$. Thus we have a path from $Q_1'$ to $Q_2'$ in $\Gamma \setminus C$. We give an example of this traversal in Figure \ref{fig:intersection}. This implies that $C$ is non-separating.

It is well known that $\mathcal{C}$ forms a generating set of $H_1(\Gamma)$, the first homology group of the surface. Now from Lemma \ref{lem:cabello} it follows that $I(C,C_l) \equiv 1 (\mod 2)$ for some $l \in [2g]$, which is a contradiction.

\end{proof}

\begin{lemma}
\label{lem:weight}
Let $G$ be a graph in {\GG} with $C$ being a simple cycle in $G$ and $E^C_1$ being the set of edges of $C$ that intersects segment $S_1$. Assume $|E^C_1|$ is even and the edges in $E^C_1$ alternate between going out and coming into the grid. Let $i_1 < i_2 < \ldots < i_{2p-1} < i_{2p}$ be the distinct indices on $S_1$ where $C$ intersects it. Then
\[
\left|circ_{w'_1}(E^C_1)\right| = \left|\sum_{k=1}^p (i_{2k} - i_{2k-1})\right|
\]
and thus non-zero unless $E^C_1$ is empty.
\end{lemma}
\begin{proof}
Let $e_j = (u_j,v_j)$ for $j \in [2p]$ be the $2p$ edges of $G$ lying on the segment $S_1$. Assume without loss of generality that the vertices $v_j$'s lie on $S_1$. Assign an orientation to $C$ such that $e_1$ is directed from $u_1$ to $v_1$. Also assume that $i_1$ is even and the circulation gives a positive sign to the edge $e_1$. Therefore $circ_{w'_1}(\{e_1\}) = -i_1$.

Now consider any edge $e_j$ such that $j$ is even. By Lemma \ref{lem:alternate}, the edge enters the segment $S_1$. Suppose $i_j$ is odd. Then consider the following cycle $C'$ formed by tracing $C$ from $u_j$ to $u_1$, without the edges $e_1$ and $e_j$ and then moving along the segment $S_1$ back to $u_j$. Since $i_j$ is odd therefore the latter part of $C'$ has odd length. Note that $C'$ need not be a simple cycle. By Lemma \ref{lem:even}, $|C'|$ is even, therefore the part of $C'$ from $u_1$ to $u_j$ also has odd length. This implies that the circulation gives a positive sign to the edge $e_j$. Therefore, $circ_{w'_1}(\{e_j\}) = i_j$. Similarly, if $i_j$ is odd, then the part of $C'$ from $u_1$ to $u_j$ will have even length. Thus the circulation gives a negative sign to the edge $e_j$ and therefore $circ_{w'_1}(\{e_j\}) = -(-i_j) = i_j$. 

If $j$ is odd, the above argument can be applied to show that $circ_{w'_1}(\{e_j\}) = -i_j$. Therefore we have,
\[
circ_{w'_1}(E^C_1) =  \sum_{k=1}^p (i_{2k} - i_{2k-1}).
\]

Now removing the assumptions at the beginning of this proof would show that the LHS and RHS of the above equation is true modulo absolute value as required.
\end{proof}

It is interesting to note here that similar method does not show that bipartite matching in non-orientable constant genus graphs is in {\SPL}. The reason is that Lemma \ref{lem:alternate} crucially uses the fact that the surface is orientable. In fact, one can easily come with counterexample to the Lemma if the surface is non-orientable.

\section{Reducing the non-orientable case to the orientable case}

Let $G$ be a bipartite graph embedded on a genus $g$ non-orientable surface. 
As a result of Theorem \ref{theorem:normal} we can assume that we are given a combinatorial embedding (say $\Pi$) of $G$ on a (non-orientable) polygonal schema,
say $\Lambda(\Gamma),$
in the normal form with $2g'$ sides.  (Here $g'$ is a function of $g.$)

Let $Y=(X_1,X_2)$ be the cyclic ordering of the labels of the sides of $\Lambda(\Gamma)$,
where $X_2$ 
is the `orientable part' and $X_1$ is the `non-orientable part'.
More precisely, 
for the polygonal schema in the normal form, we have:
$X_1$ is either $(\sigma,\sigma)$ (thus corresponds to the projective plane) or
it is  
$(\sigma,\tau,\bar{\sigma},\tau)$ (thus corresponds to the Klein bottle). 
See Figure \ref{fig:nonorientableschema}.

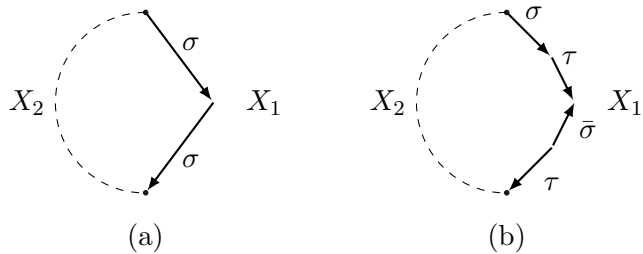
\begin{figure}[h]
\centering

\begin{tikzpicture}[scale=.6,shorten >=.35mm,>=latex]

\draw [thick] [->] (0,6) -- (1.5,4);
\draw [thick][->] (1.5,4) -- (0,2);

\draw [fill, color=black!90] (0,6) circle(0.5 mm);
\draw [fill, color=black!90] (0,2) circle(0.5 mm);

\node at (1,5.3) {$\sigma$};
\node at (1,2.7) {$\sigma$};

\draw [dashed] (0,6) arc (90:270:2);

\node [left] at (-2,4) {$X_2$};
\node [right] at (2,4) {$X_1$};
\node at (0,1) {(a)};

\draw [thick] [->] (8,6) -- (9,5);
\draw [thick][->] (9,5) -- (9.5,4);
\draw [thick][<-] (9.5,4) -- (9,3);
\draw [thick][->] (9,3) -- (8,2);

\draw [fill, color=black!90] (8,6) circle(0.5 mm);
\draw [fill, color=black!90] (8,2) circle(0.5 mm);

\node at (8.6,6) {$\sigma$};
\node at (9.4,5) {$\tau$};
\node at (9.8,3.4) {$\bar{\sigma}$};
\node at (9,2.2) {$\tau$};

\draw [dashed] (8,6) arc (90:270:2);

\node [left] at (6,4) {$X_2$};
\node [right] at (10,4) {$X_1$};
\node at (8,1) {(b)};

\end{tikzpicture}
\caption{{\bf (a)} $\Lambda(\Gamma)$ when the surface is a sum of an orientable surface and the projective plane. {\bf (b)} $\Lambda(\Gamma)$ when the surface is a sum of an orientable surface and the Klein bottle}.
\label{fig:nonorientableschema}
\end{figure}

Now let $G$ be a bipartite graph embedded on a non-orientable polygonal schema $\Lambda(\Gamma)$ with $2g'$
sides. 
We will construct a graph $G'$ embedded on an {\em orientable} polygonal schema with $4g' - 2$ sides such that 
$G$ has a perfect matching iff $G'$ has a perfect matching. 
Moreover, given a perfect matching in $G'$ one can retrieve in logspace a perfect matching in $G.$
This is illustrated in the following Theorem.

\begin{theorem}\label{theorem:nonorientable}
Let $G$ be a bipartite graph given with its embedding on a 
non-orientable polygonal schema in normal form $\Lambda(\Gamma)$, with $2g'$ sides as above. 
One can construct in logspace, another graph $G'$ together with its embedding on the polygonal schema of an orientable surface $\Gamma'$ of genus $4g' - 2$ such that:
$G$ has a perfect matching iff $G'$ has a perfect matching. Moreover, given a perfect matching in $G',$
one can construct in logspace a perfect matching in $G.$
\end{theorem}

\begin{proof}
We first show the case when $\Gamma$ is the sum of an orientable surface and a Klein bottle. Consider the polygonal schema formed by taking two copies of $\Lambda(\Gamma)$ and glueing the side $\tau$ of one copy with its partnered side $\tau$ of the other copy. We relabel the edge labelled $\sigma$ in the second copy with some unused symbol $\delta$ to avoid confusion. The entire reduction is shown in Figure \ref{fig:conversion_klein}. Let $G'$ be the resulting graph.

\begin{figure}[h]
\centering

\begin{tikzpicture}[scale=.6,shorten >=.35mm,>=latex]

\draw [] [->] (8,6) -- (9,5);
\draw [][->] (9,5) -- (9.5,4);
\draw [][<-] (9.5,4) -- (9,3);
\draw [very thick][->] (9,3) -- (8,2);

\draw [fill, color=black!90] (8,6) circle(0.5 mm);
\draw [fill, color=black!90] (8,2) circle(0.5 mm);

\node at (8.6,6) {$\sigma$};
\node at (9.4,5) {$\tau$};
\node at (9.8,3.4) {$\bar{\sigma}$};
\node at (9,2.2) {$\tau$};

\draw [dashed] (8,6) arc (90:270:2);

\node [left] at (6,4) {$X_{21}$};

\draw [][<-] (10,2) -- (11.5,2);
\draw [very thick][<-] (9,1) -- (10,2);
\draw [] [->] (8.5,0) -- (9,1);
\draw [] [<-] (8.5,-1) -- (8.5,0);
	
\draw [fill, color=black!90] (8.5,-1) circle(0.5 mm);
\draw [fill, color=black!90] (11.5,2) circle(0.5 mm);

\node [below] at (10.75,2) {$\delta$};
\node [below] at (9.75,1.7) {$\tau$};
\node [right] at (8.75,.4) {$\bar{\delta}$};
\node [right] at (8.75,-.5) {$\tau$};

\draw [dashed] (8.5,-1) arc (270:360:3);

\node [right] at (11,0) {$X_{22}$};

\node at (8,-2) {(a)};

\draw [] [->] (16,4) -- (17,3);
\draw [][->] (17,3) -- (17.5,2);
\draw [][<-] (17.5,2) -- (17,1);
\draw [][<-] (17,1) -- (16,0);
\draw [][<-] (14.5,1) -- (14,2);
\draw [][->] (14,2) -- (14.5,3);

\draw [fill, color=black!90] (16,4) circle(0.5 mm);
\draw [fill, color=black!90] (16,0) circle(0.5 mm);

\node at (16.6,4) {$\sigma$};
\node at (17.4,3) {$\tau$};
\node at (17.8,1.4) {$\bar{\sigma}$};
\node at (17,0.2) {$\delta$};
\node at (13.8,2.5) {$\bar{\delta}$};
\node at (13.8,1.4) {$\bar{\tau}$};

\draw [dashed] (16,4) arc (90:180:1.6);
\draw [dashed] (16,0) arc (270:180:1.6);

\node [left] at (15,0) {$X_{22}$};
\node [left] at (15,4) {$X_{21}$};

\node at (16,-2) {(b)};

\end{tikzpicture}
\caption{{\em Klein bottle}. {\bf (a)} The two copies of $\Lambda(\Gamma)$ with the side that is being glued shown in dark. {\bf (b)} Polygonal schema obtained after the glueing operation.}

\label{fig:conversion_klein}
\end{figure}
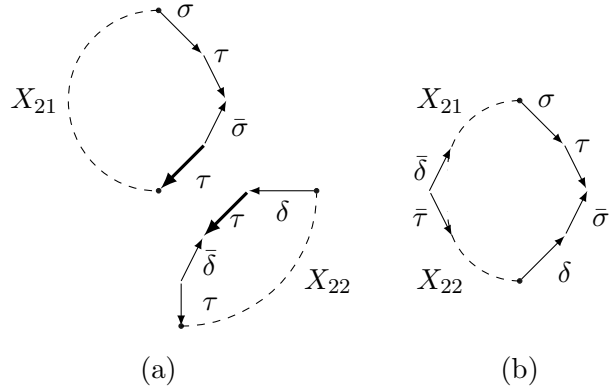

Note that the polygonal schema obtained as a result represents an orientable surface and has constantly many sides. Also every vertex and edge in $G$ has exactly two copies in $G'$ and $G'$ is also bipartite. Let $M$ be a matching in $G$. Let $M'$ be the union of the edges of $M$ from both the copies of $G$ . Its easy to see that $M'$ is a matching in $G'$. Now consider a matching $M'$ in $G'$. The {\em projection} of $M'$ to $G$ gives a subgraph of $G$ where every vertex
has degree (counted with multiplicity) exactly two. Since $G$ is bipartite, one can obtain a perfect matching within this subgraph.

Now consider the case when $\Gamma$ is the `sum' of an orientable surface and a projective plane,
i.e., following the notation above $X_1$ corresponds to the labels of a polygonal schema
for the projective plane and $X_2$ corresponds to the labels of a polygonal schema of an
orientable surface. Take two copies of $\Lambda(\Gamma)$, and glue $\sigma$ of one copy with its partner $\sigma$ in the other copy. We show this operation in Figure \ref{fig:conversion_projective}.
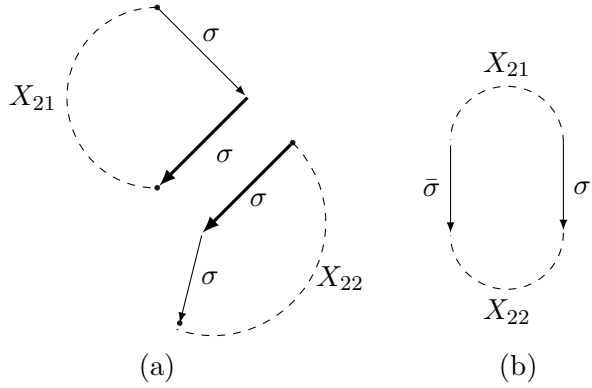
\begin{figure}[h]
\centering

\begin{tikzpicture}[scale=.6,shorten >=.35mm,>=latex]

\draw [] [->] (8,6) -- (10,4);
\draw [very thick][->] (10,4) -- (8,2);

\draw [fill, color=black!90] (8,6) circle(0.5 mm);
\draw [fill, color=black!90] (8,2) circle(0.5 mm);

\node at (9.2,5.4) {$\sigma$};
\node at (9.5,2.7) {$\sigma$};

\draw [dashed] (8,6) arc (90:270:2);

\node [left] at (6,4) {$X_{21}$};

\draw [very thick][<-] (9,1) -- (11,3);
\draw [] [<-] (8.5,-1) -- (9,1);
	
\draw [fill, color=black!90] (8.5,-1) circle(0.5 mm);
\draw [fill, color=black!90] (11,3) circle(0.5 mm);

\node [below] at (10.25,2.1) {$\sigma$};
\node [right] at (8.75,0) {$\sigma$};

\draw [dashed] (11,3) arc (45:-110:2.5);

\node [right] at (11.3,0) {$X_{22}$};

\node at (8,-2) {(a)};

\draw [][->] (17,3) -- (17,1);
\draw [][<-] (14.5,1) -- (14.5,3);

\node [right] at (17,2) {$\sigma$};
\node [left] at (14.5,2) {$\bar{\sigma}$};

\draw [dashed] (17,3) arc (0:180:1.25);
\draw [dashed] (17,1) arc (360:180:1.25);

\node [above] at (15.75,4.25) {$X_{21}$};
\node [below] at (15.75,-.25) {$X_{22}$};

\node at (16,-2) {(b)};

\end{tikzpicture}
\caption{{\em Projective plane}. {\bf (a)} The two copies of $\Lambda(\Gamma)$ with the two pair of sides that are being glued shown in dark. {\bf (b)} Polygonal schema obtained after the glueing operation.}

\label{fig:conversion_projective}
\end{figure}
The rest of the proof is similar to the Klein bottle case. 
\end{proof}

Thus we see that the non-orientable case can be reduced to the orientable case. The resulting polygonal schema need not be in the normal form. Once again we apply Theorem \ref{theorem:normal} to get a combinatorial embedding on a polygonal schema in the normal form.
\section*{Acknowledgment}

The third author would like to thank Prof. Mark Brittenham from the Mathematics department at the University of Nebraska-Lincoln, for numerous discussions that they had and for providing valuable insight into topics in algebraic topology.  

\bibliographystyle{alpha}
\bibliography{GGR}

\end{document}